\documentclass[11pt,a4paper]{article}   
\usepackage{amsthm,graphicx,amsmath,amssymb,amsfonts,framed}   
\usepackage{fullpage}
\usepackage{hyperref} 
\usepackage{color}

\title{Gauge transformations for twisted spectral triples} 
\author{Giovanni Landi and Pierre Martinetti} 
\date{}

\newtheorem{prop}{Proposition}[section]

\newtheorem{Lemma}[prop]{Lemma}
\newtheorem{cor}[prop]{Corollary}
\newtheorem{dfn}[prop]{Definition}
\newtheorem{rem}[prop]{Remark}
\newtheorem{asn}[prop]{Assumption}
\newtheorem{exl}[prop]{Example}

\numberwithin{equation}{section}

\newcommand{\nn}{\nonumber}
\newcommand{\ot}{\otimes}
\newcommand{\beq}{\begin{equation}}
\newcommand{\eeq}{\end{equation}}

\newcommand\A{{\mathcal A}}
\newcommand{\B}{{\mathcal B}}
\newcommand\E{{\mathcal E}}

\newcommand{\M}{{\mathcal M}}
\newcommand{\HH}{{\mathcal H}}
\newcommand{\cE}{{\mathcal E}}
\newcommand{\U}{{\mathcal U}}

\newcommand{\bo}{{\boldsymbol{\omega}}}

\newcommand{\I}{\mbox{\rm 1\hspace {-.6em} l}}

\newcommand{\N}{{\mathbb N}}
\newcommand{\R}{{\mathbb R}}
\newcommand{\C}{\mathbb C}

\newcommand{\cinf}{{C^\infty(\M)}}
\newcommand{\ds}{{\slash \!\!\!\partial}}
\newcommand{\ii}{{\rm i}}

\def\hs#1#2{\left\langle #1,#2\right\rangle}  %
\def\lhs#1#2{{_\bullet\!\!}\left\langle #1,#2\right\rangle}
\def\rhs#1#2{\left\langle #1,#2\right\rangle\!\!{_\bullet}}

\DeclareMathOperator{\Aut}{Aut}
\DeclareMathOperator{\Ad}{Ad}
\DeclareMathOperator{\id}{id}

\DeclareMathOperator{\Mat}{Mat}
\DeclareMathOperator{\Hom}{Hom}
\DeclareMathOperator{\End}{End}

\begin{document}
\maketitle

\begin{abstract}
It is extended to twisted spectral triples the fluctuations 
of the metric as bounded perturbations of the Dirac operator that arises when
a spectral triple is exported between Morita equivalent
algebras; as well as gauge transformations which are
obtained by the action of the unitary endomorphisms of the module
implementing the Morita equivalence. 
It is firstly shown that the twisted gauged  Dirac operators,
previously introduced to generate an extra scalar field in the spectral description of the 
standard model of elementary particles, in fact follow from Morita equivalence between twisted spectral triples.
The law of transformation of the gauge
potentials turns out to be twisted in a natural way. 
In contrast with the non-twisted case, twisted fluctuations do
not necessarily preserve the self-adjointness of the Dirac operator.
For a self-Morita equivalence conditions  are obtained in order to maintain self-adjointness, 
that are solved explicitly for the minimal twist of a Riemannian manifold.
\end{abstract}

\setcounter{tocdepth}{2}
\tableofcontents

\parskip = 1ex

\thispagestyle{empty}

\vfill


\noindent
-------- -------- -------- -------- -------- \\
GL, Universit\`{a} di Trieste, Trieste, Italy and I.N.F.N. Sezione di Trieste, Trieste, Italy; \\
PM, Universit\`{a} di Genova, Genova, Italy; \\
emails: landi@units.it, martinetti@dima.unige.it \hfill \date{v2 4 May 2018; v1 18 April 2017}
 
 \section{Introduction}
The gauge bosons of the standard model of elementary particles are described by (quantum) fields that,
from a mathematical view-point, are connections $1$-forms for a bundle over a (four dimensional) spin
manifold $\M$, with structure (gauge) group $U(1)\times SU(2)\times SU(3)$.
 Noncommutative geometry provides a framework to put the Higgs field
 on the same footing --- that is as a
 connection $1$-form --- or more precisely as the component of a connection
 $1$-form in the noncommutative (discrete) part of the geometry. 
For this to make sense, one needs a notion of
 connection extended beyond the usual manifold case, to 
 the noncommutative setting. 

In Connes approach \cite{Connes:1994kx}, this is done starting with a spectral triple $(\A, \HH, D)$
 where $\A$ is an involutive algebra acting by bounded operators on a
 Hilbert space $\HH$, and the \emph{Dirac operator} $D$ is a densely defined self-adjoint
 operator on $\HH$ with compact resolvent, such that the commutator{\footnote{As usual, when there is no
    risk of confusion we identify an
    element $a$ of $\A$ with its representation $\pi(a)$ as a bounded
    operator on $\HH$.}}
\begin{equation}
\delta(a):=[D, a]
\label{eq:51}
\end{equation}
is
 bounded for any $a$ in $\A$ (or in a dense subalgebra). The noncommutative analogue of the
 module of sections of a vector or tensor bundle is a $\A$-module $\E$ with some properties. 
Gauge fields are given by an $\Omega$-valued 
 connection on $\E$, where $\Omega$ is a $\A$-bimodule of $1$-forms.  
  A natural choice of these, associated with the derivation \eqref{eq:51}, is 
 the $\A$-bimodule
 \begin{equation}
\label{eq:71}
  \Omega^1_D(\A):=\Big\{\sum\nolimits_j a_j [D, b_j], \; a_j, b_j\in\A\Big\}.
\end{equation}

The simplest choice for $\E$ is the algebra $\A$ itself. A connection is
then encoded fully in a self-adjoint element $\omega$ in  $\Omega^1_D(\A)$. The later acts
on the Hilbert space $\HH$, so that $D+\omega$ makes sense as an operator on $\HH$.
By taking into account more structure,
in particular the real structure $J$, one refines the above
definition and defines the \emph{gauged Dirac operator}
as{\footnote{Usually one denotes by $A$ a
    self-adjoint element of $\Omega^1_D(\A)$ considered as a gauge connection.
Here we use $\omega$ instead,  in order to avoid a profusion of  symbols ``$A$'' .}} 
\begin{equation}
  \label{eq:77}
 D_\omega:= D+ \omega + \epsilon' J \omega J^{-1}
\end{equation}
where $\epsilon' = \pm 1$ as dictated by the $KO$-dimension of the spectral triple. 
This is an operator on $\HH$,that has all the properties required to make $(\A, \HH, D_\omega)$
 a spectral triple. The substitution of $D$ by $D_\omega$ is a \emph{fluctuation of the metric}, the latter `associated' to the starting $D$.

When applied to the
spectral triple of the standard model, these fluctuations generate the gauge fields of
the electroweak and strong interactions, together with the Higgs
field \cite{Chamseddine:2007oz}. 
There is however a part $D'$ of the corresponding Dirac operator which does not fluctuate, that is \begin{equation}
[D', a]=0 \text{ for any  }a\in \A.
\label{eq:81}
\end{equation}
This point was not relevant until the recent discovery of
the Higgs boson. The prediction for its mass 
coming from noncommutative geometry
turned out \emph{not} to be in agreement with the experimental result. As a way out, one
 turns the component of $D'$ (which was taken to be a constant parameter $\nu\in\C$) into
a field $\sigma\in\cinf$. Doing so,
one  introduces a new scalar field in the standard model, that eliminates some instability in the Higgs
potential, and provides a new parameter allowing one to fit the mass of
the Higgs \cite{Chamseddine:2012fk}.

The substitution $\nu\rightarrow
\sigma$ does not follow from an ordinary fluctuation of the
metric. Nevertheless, it may be obtained in a similar manner if one relaxes one of the defining condition of a spectral triple, --- the first order condition.
This proposal has been developed 
in \cite{Chamseddine:2013fk,Chamseddine:2013uq}, and the
phenomenological consequences have been investigated in \cite{Chamseddine:2015aa}.
An alternative approach, following the ``grand symmetry model'' of \cite{Devastato:2013fk}, 
has allowed in \cite{buckley} to generate the field $\sigma$ within the framework of twisted spectral
triples~\cite{Connes:1938fk}: the field
$\sigma$ is obtained as a twisted version of a fluctuation of
the metric, with a twisted first-order condition. A twisted
fluctuation of the metric comes from substituting in the forms \eqref{eq:71} the commutator $[D,a]$ with a twisted commutator $[D, a]_\rho := D a - \rho(a) D$,
using an automorphism $\rho$ of $\A$, resulting into a bimodule
\begin{equation}
  \label{eq:2}
  \Omega^1_D(\A, \rho):= \Big\{ \sum\nolimits_j a_j [D, b_j]_\rho,\; a_j, b_j\in\A\Big\}
\end{equation}
The \emph{twisted-gauged Dirac operator} is then defined as
\begin{equation}
  \label{eq:76}
  D_{\omega_\rho} := D +\omega_\rho + \epsilon' J \omega_\rho J^{-1}
\end{equation}
where $\omega_\rho \in \Omega^1_D(\A, \rho)$ is a twisted $1$-form such that the resulting operator \eqref{eq:76} is self-adjoint.

Twisted spectral triples and twisted $1$-forms were introduced in \cite{Connes:1938fk} 
to deal with type III factors. In \cite{Landi:2015aa} we extended the construction to encompass the
real structure $J$, and showed that many properties of metric fluctuations
still make sense in the twisted case. In particular: 
\begin{itemize}
\item 
Given a twisted spectral triple $(\A, \HH, D; \rho)$ and a
  twisted-gauged Dirac operator $D_{\omega_\rho}$, the data 
  $(\A, \HH, D_{\omega_\rho}; \rho)$ is a real
  twisted spectral triple with the same real structure and $KO$-dimension;

\item 
Twisted fluctuations form a monoid:{\footnote{There is a misprint in
    the statement of this property in \cite[Prop. 2.7]{Landi:2015aa}:
$D_\rho$ in (2.30) there should be $D$.}} the twisted fluctuation
  $D_{\omega_\rho} + \omega'_\rho + \epsilon' J \omega'_\rho J^{-1}$ of $D_{\omega_\rho}$ is the twisted
  fluctuation $D + \omega''_\rho + \epsilon' J \omega''_\rho J^{-1}$ of $D$ with
  $\omega''_\rho = \omega_\rho + \omega'_\rho$.
 \end{itemize}

\noindent
However, important aspects and consequences of fluctuating the metric are
 yet to be understood for the twisted case. In particular:

\begin{itemize}
\item 
Usual fluctuations appear as a particular case of a general construction of 
exporting a spectral triple $(\A, \HH, D)$ to a Morita equivalent algebra $\B$. The operator
\eqref{eq:77} is obtained as the covariant derivative on
the bundle $\E$ that implements a Morita equivalence of $\A$ with itself. 
The twisted fluctuations in \eqref{eq:76} 
mimic the expression for the non-twisted
case, but their possible interpretation in terms of Morita
equivalence has not been addressed.  

\item 
Is there an interpretation of the bimodule $\Omega^1_D(\A, \rho)$ 
as a module of connection $1$-forms?

\item 
What is a gauge transformation in the twisted context?
 \end{itemize}
 
In this paper, we show that Morita equivalence is directly implemented
for twisted spectral triples. The twisted-gauged Dirac operator  $D_{\omega_\rho}$  is --- up to an endomorphism --- a covariant operator associated
to a connection on the algebra $\A$ thought of as an $\A$-bimodule.
This result is obtained in \S\ref{subsec:twistfluct} by viewing $\A$ first as a right $\A$-module
(Corollary~\ref{corro:right}), then as a left $\A$-module (Corollary~\ref{corro:left}),
and finally as a bimodule, taking into account the real structure  (proposition
\ref{prop:twistMorita}). 
In \S\ref{sec:twisgaugetransform} we deal with gauge
transformations. These are implemented as in the non-twisted case by
the action of some unitary endomorphism $u$, the only
dif and only if erence being that 
the law of transformation of gauge potential has to be twisted (Proposition~\ref{prop:twistedgaugefluct}). We
also show in Proposition~\ref{prop:twistedunitaries} that the twisted gauged
Dirac operator is obtained by the twisted adjoint action of the operator $\Ad(u)$. This raises the question of the self-adjointness of the
gauged-twisted Dirac operator, which is investigated in  \S\ref{sec:selfadjoint}. We work out in 
Proposition~\ref{prop:conditionsselfadjoint} some conditions on the
unitary $u$ guaranteeing that this
self-adjointness is preserved. These conditions are solved for
the case of minimal twist of a manifold
(Proposition~\ref{prop:minimaltwist}). Interestingly, we obtain 
other solutions than the obvious ones (that is the unitaries $u$ invariant under the
twist).  Before that, we begin in  \S\ref{sec:reltsp} with
some recalling of twisted spectral triples.

In \cite{BCDS16} there is a modified definition of a real spectral triple, in which only the reality structure is generalised while remaining in the framework of usual spectral triples (that is no twisted commutators between the Dirac operator and algebra elements). It is shown there that this  allows for fluctuations of Dirac operators, which do not change the bimodule of one forms. 

\section{Twisted Real spectral triples}
\label{sec:reltsp}

This section collects well known material on  and properties of real twisted spectral triples.

A twisted spectral triple is the datum $(\A, \HH, D)$ of an
involutive algebra  $\A$ acting via a representation $\pi$ on a Hilbert space $\HH$, with $D$ an
operator on $\HH$ having compact resolvent (or with a similar condition when $\A$ is not unital), together with an automorphism $\rho$ of $\A$, such that the twisted commutator
\begin{equation}
  \label{eq:1}
  [D, a]_\rho := D a - \rho(a) D
\end{equation}
is bounded for any $a$ in $\A$.  It is graded if there is
a grading $\Gamma$ of $\HH$, that is an operator such that 
$\Gamma=\Gamma^*$, $\Gamma^2=\I$, that commutes with $\A$ and anticommutes
with $D$. 

The real structure is an antilinear operator $J$ such that
\begin{equation}
  \label{eq:4}
  J^2 =\epsilon \I,\quad JD= \epsilon' DJ,\quad J\Gamma = \epsilon''
  \Gamma J
\end{equation}
where the sign $\epsilon, \epsilon', \epsilon'' \in\left\{1,-1\right\}$ define
  the so called $KO$-dimension of the spectral triple.
The operator $J$ allows one to define a bijection between $\A$ and the opposite
  algebra $\A^\circ$,
  \begin{equation}
    \label{eq:24}
    a^\circ := Ja^*J^{-1},
  \end{equation}
which is used to implement a right $\A$-module
structure on $\HH$   
\begin{equation}
  \label{eq:35}
  \psi a := a^\circ \psi \quad \forall \psi\in\HH, a\in\A.
\end{equation}
This right
action of $\A$ is asked to commute with the left action (the
order-zero condition), 
  \begin{equation}
    \label{eq:12}
    [a, Jb^*J^{-1}]=0 \quad \forall a, b\in\A,
  \end{equation}
thus turning $\HH$ into a $\A$-bimodule. In addition, one requires a twisted first-order condition 
\cite{Landi:2015aa}:
 \begin{equation}
   \label{eq:13}
   [[D,a]_\rho, Jb^*J^{-1}]_{\rho_0}= 0 \quad \forall a, b\in \A
 \end{equation}
where
$\rho^\circ$ is the image of $\rho$ under the isomorphism between
$\text{Aut}(\A)$ and $\text{Aut}(\A^\circ)$ given by
\begin{equation}
  \label{eq:196}
  \rho\mapsto \rho^\circ \quad \text{ with } \quad\rho^\circ(a^\circ):= (\rho^{-1}(a))^\circ.
\end{equation}
 This choice of isomorphism is dictated by the requirement made in
 \cite{Connes:1938fk} that the twisting automorphism, rather than being a $*$-automorphism, it satisfies the condition:
\begin{equation}
  \label{eq:170}
  \rho(a^*)= (\rho^{-1}(a))^*.
\end{equation} 
Equation \eqref{eq:196} thus guarantees that ``the automorphism commutes
with the real structure'', since one has:
\begin{equation}
\rho^\circ(Jb^*J^{-1}) = \rho^\circ(b^\circ) = (\rho^{-1}(b))^\circ= J (\rho^{-1}(b))^* J^{-1} =  J \rho(b^*)J^{-1} .
\label{eq:5bis}
\end{equation}

\begin{dfn}
\label{deftwist}
  A twisted spectral triple $(\A, \HH, D; \rho)$ together with a
  grading $\Gamma$,  a real structure $J$ satisfying \eqref{eq:4} as
  well as the order zero condition \eqref{eq:12}, and the twisted first-order condition \eqref{eq:13} is called a \emph{real twisted
    spectral triple}.
\end{dfn}
\noindent For $\rho$ the identity automorphism, one gets back
the usual notion of a real spectral triple.

The set of twisted $1$-forms is the $\A$-bimodule $\Omega^1_D(\A, \rho)$ 
defined in \eqref{eq:2} with product
\begin{equation}
  \label{eq:15}
  a\cdot \omega_\rho\cdot b = \rho(a)\, \omega_\rho\, b \;\quad \forall a, b \in \A,\;\,
  \omega_\rho\in \Omega^1_D(\A \rho).
\end{equation}
 The left action of $\A$ is twisted by
$\rho$ to guarantee the twisted commutator
\begin{equation}
\delta_\rho( \, \cdot \, ) := [D, \, \cdot \, ]_\rho\label{eq:91}
\end{equation}
be a derivation of $\A$ in
$\Omega^1_D(\A, \rho)$, that is (cf. \cite{Connes:1938fk}) 
\begin{equation}
  \label{eq:96}
  \delta_\rho(ab) = \rho(a) \cdot\delta_\rho(b) + \delta_\rho(a)\cdot b.
\end{equation}
Thus $\Omega^1_D(\A, \rho)$ is the $\A$-bimodule
generated by $\delta_\rho$; and it acts as bounded operator on $\HH$,
since so do both $\A$ and $[D, \A]_\rho$. It is worth stressing a dif and only if erence between the right and left action of 
$\A$ on $1$-forms when acting on $\HH$. By the very definition in \eqref{eq:15}, one has 
\beq \label{eq:1190}
  (\omega_\rho\cdot a)\psi  = \omega_\rho a\psi= \omega_\rho (a\psi),\eeq
while
\beq \label{eq:1160}
(a\cdot\omega_\rho)\psi =\rho(a) \omega_\rho\psi \neq a(\omega_\rho\psi) .
\eeq

\section{Twisted fluctuation by Morita equivalence}
\label{subsec:twistfluct}

In the non-twisted case, the fluctuations of the metric
arise as a way to export a spectral triple $(\A, \HH, D)$ to an
algebra $\B$ which is Morita
equivalent to $\A$, in a way compatible with the real structure. An important role is played by a connection on a module
that is moved to the Hilbert space (\S \ref{sec:connec}) thus resulting into a gauged Dirac
operator (\S\ref{ntcm}). This construction is extended to the
twisted situation in \S\ref{subsec:lift}-\ref{subsec:bimodandreal}. The
main result is Proposition~\ref{prop:twistMorita}, which shows that
the twisted-gauged Dirac operator \eqref{eq:76} is obtained by Morita
equivalence,  in a way  similar to the one for the usual gauged Dirac operator \eqref{eq:77}.

\subsection{Moving connections to Hilbert spaces}\label{sec:connec}
We recall
how an $\Omega$-valued connection on a right (or left) $\A$-module $\E$
yields a map $\nabla$ on $\E\otimes_\C \HH$ (or $\HH\otimes_\C\E$), when both the
$\A$-bimodule $\Omega$ and the algebra $\A$ act on $\HH$. This map
does not pass to the tensor product $\E\otimes_\A \HH$ (or
$\HH\otimes_\A\E$).  We get in
Propositions~\ref{Prop:lebnizright} and \ref{prop:leibnizleft} 
compatibility conditions between the  actions of $\A$ and $\Omega$ which guarantees that this
lack of $\A$-linearity of $\nabla$ is captured by the derivation $\delta$ that
generates $\Omega$. 

 A derivation of an algebra $\A$
 with value in a $\A$-bimodule $\Omega$ is a map $\delta: \A \to
 \Omega$ such that
 \begin{equation}
   \label{eq:113}
   \delta(ab) =\delta(a)\cdot b + a\cdot \delta(b)
 \end{equation}
where $\cdot$ denotes the right and left $\A$-module structures of $\Omega$. 
An $\Omega$-valued 
connection on a right $\A$-module $\E$ is a map $\nabla: \E\to \E\otimes_\A \Omega$ satisfying the Leibniz rule
\begin{equation}
  \label{eq:14}
  \nabla (\eta a) -\nabla(\eta)\cdot a = \eta\otimes \delta(a) \quad
  \forall \eta\in\E,\, a\in\A,
\end{equation}
where the right action of $\A$ on  $\E\otimes_\A \Omega$ comes
from the right module structure of $\Omega$:  
\begin{equation}
  \label{eq:18}
  (\eta \otimes \omega)\cdot a := \eta \otimes (\omega\cdot a) \quad
  \forall \eta\in\E,\; \omega\in\Omega.
\end{equation}
%
%

When both $\A$ and $\Omega$ acts (on the left) on a Hilbert space $\HH$, we use the connection $\nabla$ 
to define an operator (still denoted $\nabla$) from $\E\otimes_\C
  \HH$ to itself.
To this end, it is useful to use a Sweedler-like notation: for any $\eta \in \cE$ we write
\beq\label{sn}
\nabla (\eta) = \eta_{(0)} \ot \eta_{(1)} \quad
\quad\eta_{(0)} \in \cE, \;\eta_{(1)} \in \Omega
\eeq
where a summation is understood. 
By the action of $\Omega$ on $\HH$, there is a natural map 
\beq\label{eoact}
\E\otimes_\C \Omega \times \HH \to \E\otimes_\C\HH , \qquad (\eta \ot \omega) \psi = \eta \ot (\omega \psi),
\eeq
that induces a map
\beq
\label{eq:115bis}
\nabla : \cE \otimes_\C\HH \, \to \, \cE \ot_\C \HH
\eeq
defined by
\begin{equation} \label{eq:115}
  \nabla(\eta \otimes \psi) := \left(\eta_{(0)} \ot \eta_{(1)}\right)
  \psi =\eta_{(0)} \ot (\eta_{(1)} \psi) \quad\quad \forall
  \eta\in\E,\, \psi\in\HH.
\end{equation}
Somewhat abusing notation, this is often denoted as
$\nabla(\eta)\psi$. 

This map cannot be extended to the tensor product $\E\otimes_\A \HH$ over
$\A$ because there is no reason that $\nabla(\eta a)\psi - \nabla(\eta)a\psi$
vanishes. However, this incompatibility  is captured by the
derivation $\delta$, providing the actions of $\Omega$ and $\A$ on $\HH$ are compatible.
\begin{prop}
\label{Prop:lebnizright}
  If the (left) actions of $\Omega$ and $\A$ on $\HH$ are such that 
 \begin{equation}
   \label{eq:129}
(\omega\cdot a)\,\psi = \omega(a\psi),
 \end{equation}
then the map $\nabla$ in \eqref{eq:115} satisfies the Leibniz rule 
\begin{equation}
  \label{eq:122}
   \nabla(\eta a)\psi - \nabla(\eta)a\psi = \eta\otimes \delta(a)\psi \quad\quad
  \forall \eta\in\E,\, a\in\A,\, \psi\in\HH.
\end{equation}
\end{prop}
\begin{proof}
In Sweedler notations, the Leibniz rule \eqref{eq:14} reads
\beq\label{eq:14b}
(\eta a)_{(0)} \ot (\eta a)_{(1)} - \eta_{(0)} \ot (\eta_{(1)} \cdot a ) = \eta\ot \delta(a) .
\eeq
Hence, using condition \eqref{eq:129} in the second equality
\begin{align*}
\nabla(\eta a)\psi - \nabla(\eta) a\psi & = (\eta a)_{(0)} \ot (\eta a)_{(1)} \psi - 
\eta_{(0)} \ot \eta_{(1)} (a \psi) \\
& = (\eta a)_{(0)} \ot (\eta a)_{(1)} \psi - \eta_{(0)} \ot (\eta_{(1)} \cdot a) \psi \\
& = \left( (\eta a)_{(0)} \ot (\eta a)_{(1)} - \eta_{(0)} \ot (\eta_{(1)} \cdot a ) \right) \psi \\
& = \left(\eta\ot \delta(a) \right) \psi .
\end{align*}
Equation \eqref{eq:122} follows by \eqref{eoact}.
\end{proof}

Similarly, an $\Omega$-valued connection on a left $\A$-module $\E$ is a map 
$\nabla:\E\to\Omega\otimes_\A \E$ such that 
\begin{equation}
  \label{eq:58}
\nabla (a\eta) -   a\cdot\nabla(\eta)= \delta(a)\otimes \eta  \quad \quad \forall a\in\A ,\, \eta\in\E ,
\end{equation}
with left multiplication by $\A$ on $\Omega\otimes_\A\E$ coming from the left module structure
of $\Omega$, 
\begin{equation}
  \label{eq:67}
 a\cdot  (\omega\otimes \eta):=  (a\cdot \omega)\otimes \eta \quad\quad
  \forall \eta\in \E,\,\omega\in\Omega.
\end{equation}
For any $\eta \in \cE$ we shall now write with Sweedler-like notation
\beq
\nabla (\eta) = \eta_{(-1)} \ot \eta_{(0)} \quad \eta_{(0)} \in \cE, \,\eta_{(-1)} \in \Omega.
\eeq
%
When $\A$ acts (on the right) and $\Omega$ acts (on the
  left) on a Hilbert space $\HH$, the connection on the left module
defines a map similar to the one in \eqref{eq:115bis} with minimal changes. The map
\beq\label{eoact1}
\HH \times \Omega \otimes_\C \E\to  \HH \otimes_\C \E , \qquad (\psi) (\omega \ot \eta) =  (\omega \psi) \ot \eta,
\eeq
induces now a map $\nabla: \HH\otimes_\C\E\to \HH\otimes_\C\E$
\begin{equation} \label{eq:115b}
  \nabla(\psi\otimes \eta):= (\psi) (\eta_{(-1)} \ot \eta_{(0)}) = (\eta_{(-1)}\psi) \ot \eta_{(0)} \quad\quad \forall \eta\in\E,\, \psi\in\HH.
\end{equation}
We denote this map as $\psi\nabla(\eta)$.
Again, the obstruction to extend \eqref{eq:115b} to $\HH\otimes_\A\E$ is captured by the derivation $\delta$, if the
actions of $\Omega$ and $\A$ are compatible.
\begin{prop}
  \label{prop:leibnizleft}
If the left action of $\Omega$ and the right action of $\A$ on $\HH$ are such that
\begin{equation}
  \label{eq:130}
(a\cdot \omega)\psi = \omega(\psi a)\end{equation}
then the map $\nabla$ in \eqref{eq:115b} satisfies the Leibniz
rule
\begin{equation}
  \label{eq:123b}
\psi\nabla(a\eta)  -  \psi a\nabla(\eta) = \delta(a)\psi \ot \eta.
\end{equation}
\end{prop}
\begin{proof}
In Sweedler notations, the left Leibniz rule \eqref{eq:58} becomes
\beq\label{eq:14c}
(a \eta)_{(-1)} \ot (a \eta)_{(0)}  - (a \cdot \eta_{(-1)}) \ot \eta_{(0)} = \delta(a) \ot \eta .
\eeq
Using condition \eqref{eq:130} in the second equality: 
\begin{align*}
\psi\nabla(a \eta) - \psi a\nabla(\eta) 
& = (a \eta)_{(-1)} \psi \ot (a \eta)_{(0)} -  \eta_{(-1)} (\psi a) \ot \eta_{(0)} \\
& = (a \eta)_{(-1)} \psi \ot (a \eta)_{(0)} -  (a \cdot \eta_{(-1)}) \psi \ot \eta_{(0)} \\
& = (\psi) \big( (a \eta)_{(-1)} \ot (a \eta)_{(0)} - (a \cdot \eta_{(-1)}) \ot \eta_{(0)} \big)    \\
& = (\psi) \big(\delta(a) \ot \eta\big).  
\end{align*}
Equation \eqref{eq:123b} follows by \eqref{eoact1}.
\end{proof}


\subsection{The non-twisted case}\label{ntcm}
For completeness, the details of the construction are reported 
in \S\ref{subsec:gauge-non-twist}, while here we recall the important steps.
Following \cite{Connes-Marcolli}, a fluctuation from $D$ to the gauged operator $D_\omega$ given in
\eqref{eq:77} with $\omega\in\Omega^1_D(\A)$, is seen as a two steps process:  starting with a real spectral triple $(\A,
\HH, D), J$  one first implements a self-Morita
equivalence of $\A$ using as module the algebra itself, viewed as a
right $\A$-module $\E_R=\A$. This yields a new spectral triple
$(\A, \HH, D+\omega)$ with $\omega\in\Omega^1_D(\A)$. However this is not a
real spectral triple. To correct this lacking, one repeats the operation using still the algebra as a module, but
this time as a left $\A$-module $\E_L=\A$. The iteration yields the real spectral
triple $(\A, \HH, D_\omega=D+\omega+ J\omega J^{-1})$.   

Recall that at a first level, the algebra $\B$ is 
Morita equivalent to the unital algebra $\A$ if it is isomorphic to the
algebra of $\A$-linear (adjointable) endomorphisms of a finite projective (right say)
$\A$-module $\E_R$, that is 
$ \B \simeq \End_\A(\E_R)$.
Assuming $\E_R$ is a hermitian module, that is it carries an $\A$-hermitian structure, one use this structure to make the tensor product 
$$
\HH_R= \E_R\otimes_\A\HH
$$
into a Hilbert space (with Hilbert product recalled in \eqref{eq:32}), on which the algebra 
$\B$ acts on the left in a natural manner. The ``simplest'' action of $D$ on
$\HH_R$, that is
\begin{equation}
  \label{eq:42}
  D_R(\eta\otimes\psi):= \eta\otimes D\psi \quad \forall \eta\in\E_R, \psi\in\HH
\end{equation}
is not compatible with the tensor product of $\A$; it needs be
corrected by a connection $\nabla$ with value in $\Omega^1_D(\A)$.
The resulting covariant derivative, $D_R:=D_R +\nabla$, 
is well defined on $\HH_R$. With the notation \eqref{sn} for the 
connection this operator can be written as
\beq\label{rgdo}
D_R(\eta \otimes \psi) = \eta\otimes D\psi + \eta_{(0)} \ot (\eta_{(1)} \psi) \quad\quad \forall
  \eta\in\E,\, \psi\in\HH.
\eeq
When $\nabla$ is self-adjoint, the datum $({\B}, \HH_R, D_R)$ is a
spectral triple \cite{BMvS16}. It could be said to be `\emph{Morita equivalent'} to the starting $(\A, \HH, D)$.  
However, when $(\A, \HH, D)$ is a real spectral triple,  its
real structure $J$ is not a real structure for $({\B}, \HH_R, D_R)$.
 To cure that, one uses the right action
  \eqref{eq:35} of $\A$
on $\HH$ to fluctuate a second time, using a left module $\E_L$ endowed with an $\A$-hermitian structure. One considers the Hilbert space
$$  
\HH_L := \HH\otimes_\A \E_L
$$
on which the simple operator, 
\begin{equation}
  \label{eq:92}
  D_L(\psi\otimes\eta) := D\psi\otimes\eta , 
\end{equation}
is now made compatible with the tensor product thanks to a (left) connection $\nabla^\circ$. 
The resulting covariant operator $D_L+\nabla^\circ$ is well defined on $\HH_L$, with an expression similar to that in \eqref{rgdo}. 

Combining the two constructions, one obtains an
operator $D' = D + \nabla + \nabla^\circ$ on a Hilbert space $
\HH_{RL} = \E_R\otimes_\A\HH \otimes_\A \E_L$. 
The real structure requires that $\nabla=\nabla^\circ$. 

For a self Morita equivalence of $\A$,  that is $\B \simeq\A$, one
gets that $D'$ is the gauged operator $D_\omega$ defined in
\eqref{eq:77}, for a self-adjoint element $\omega$ in  $\Omega^1_D(\A)$.
Thus, the spectral triple $(\A, \HH, D_\omega)$
obtained by fluctuation of the metric is self-Morita equivalent to the starting one $(\A,\HH, D)$.

\subsection{Lifting automorphisms} 
\label{subsec:lift}

To adapt the construction above to the twisted case,  one needs some
action of $D$ on $\HH_R$ and $\HH_L$ whose non-compatibility with the
tensor product can be
corrected by derivations with value in $\Omega^1_D(\A, \rho)$. Such
operators are obtained in Propositions~\ref{prop.rightfluct} and
\ref{prop.leftfluct} below, by twisting  the operators $D_R$ and
$D_L$ of \eqref{eq:42} and \eqref{eq:92} with a lift of the automorphism $\rho$ to the module.

\begin{asn}\label{def:end-rho}
With a right $\A$-module $\E$ (resp. a left $\A$-module $\E$),
the automorphism $\rho$ can be lifted to $\E$ in the sense that 
there is an invertible linear map 
$\widetilde{\rho}:\E\to\E$ such that, 
\beq
\widetilde{\rho} ( \eta a) = \eta \, \rho(a) \quad\quad \mbox{resp.} \quad\quad 
\widetilde{\rho} ( a \eta) = \rho(a) \, \eta \qquad  \forall \eta\in\E, \, a\in\A . 
\eeq
\end{asn}

\medskip
\begin{exl}\label{def:end-rho-b}
With a right $\A$-module $\E_R=p\A^N$ for a projection $p = (p_{jk}) \in \Mat_N(\A)$, 
which is invariant for $\rho$, that is $\rho(p_{jk}) = p_{jk}$, 
one defines the action of $\rho\in\Aut(\A)$ on $\E$ by
\begin{equation}
\widetilde{\rho}(\eta) :=p\, \left(\begin{array}{c}
                                                       \rho(\eta_1)\\ \vdots \\
                                                       \rho(\eta_N)\end{array}\right)
\qquad \textup{for} \quad \eta = p \left(\begin{array}{c}
                                                       \eta_1\\ \vdots \\
                                                       \eta_N\end{array}\right)\in
                                                   \E_R,
                                                   \quad \eta_j \in\A. 
\label{eq:85}
\end{equation}
Similarly, the action of $\rho$ on a left $\A$-module $\E_L=\A^Np$ with an invariant projection is
given by
\begin{equation}
  \label{eq:124}
\widetilde{\rho}(\eta) := (\rho(\eta_1), \hdots, \rho(\eta_N))\,p 
\qquad \textup{for} \quad \eta=\left( \eta_1, \hdots, \eta_N \right) p \in\E_L, \quad \eta_j\in\A. 
\end{equation}
In particular, for the trivial module $\E_R=\E_L\simeq \A$ (that is $p=\I$) which
is the case relevant for the self Morita equivalence, then $\widetilde{\rho}$ is simply the automorphism $\rho$.
\end{exl}

\subsection{Morita equivalence by right module}
\label{subsec:rightmodule}

 We first investigate the implementation of Morita equivalence for a
 twisted spectral triple $(\A, \HH, D; \rho)$ using a hermitian finite
 projective right $\A$-module $\E_R$ (definitions are in \S\ref{subsec:Moritarightnontwist}).

Consider the Hilbert space $\HH_R= \E_R\otimes_\A\HH$.
As  ``natural action" of  $D$ on $\HH_R$, one considers the composition 
of $D_R$ in \eqref{eq:42}  with the endomorphism
$\rho$ of Assumption \ref{def:end-rho}, that is, 
\begin{equation}
  \label{eq:78bis}
((\widetilde{\rho}\otimes \I)\circ
 D_R)(\eta\otimes \psi) = \widetilde{\rho}(\eta)\otimes D\psi \quad \forall
 \eta\in\E_R,\, \psi\in\HH.
\end{equation}
\noindent This is not compatible with the
tensor product over $\A$ since 
\begin{align} 
((\widetilde{\rho}\otimes \I)\circ D_R)(\eta a\otimes \psi) -((\widetilde{\rho}\otimes \I)\circ D_R)(\eta \otimes a\psi) 
&= \widetilde{\rho}(\eta a)\otimes  D\psi - \widetilde{\rho}(\eta)\otimes Da\psi  \nonumber \\
&= \widetilde{\rho}(\eta)\rho(a)\otimes D\psi - \widetilde{\rho}(\eta)\otimes Da\psi \nonumber \\
&= \widetilde{\rho}(\eta)\otimes \rho(a)D\psi -\widetilde{\rho}(\eta)\otimes Da\psi \nonumber  \\
& =-\widetilde{\rho}(\eta)\otimes [D, a]_\rho \psi 
\label{eq:86last} 
\end{align}
has no reason to vanish. The r.h.s. of \eqref{eq:86last} is 
--- up to a twist --- the action on $\HH_R$ of the derivation \eqref{eq:91}. So to turn \eqref{eq:78bis}
into a well defined operator on $\HH_R$, one should proceed as in the non
twisted case and add the action of a
connection.


\begin{prop}
\label{prop.rightfluct}
Let $\nabla$ be an $\Omega^1_D(\A, \rho)$-valued connection on $\E_R$. Then the operator
  \begin{equation}
\widetilde D_R:= (\widetilde{\rho}\otimes \I)\circ(D_R + \nabla)
\label{eq:90}
\end{equation}
is well defined on $\HH_R$, with $\nabla$  the
map on $\E_R\otimes \HH$ induced by the connection, as in \eqref{eq:115}.
\end{prop}
\begin{proof}
The module law \eqref{eq:15} guarantees
that
  $(\omega_\rho\cdot a)\psi = \omega_\rho(a\psi)$,
so that by Proposition \ref{Prop:lebnizright} the map $\nabla$
satisfies the Leibniz rule
\begin{equation}
\nabla(\eta a)\psi - \nabla(\eta)a\psi =
\nabla(\eta a\otimes\psi) - \nabla(\eta\otimes a\psi) =\eta \otimes \delta_\rho(a)\psi. 
\end{equation}
Therefore
\begin{equation}
  \label{eq:95bis}
 ( (\widetilde{\rho}\otimes\I)\circ \nabla)(\eta a\otimes \psi) - ( (\widetilde{\rho}\otimes\I)\circ \nabla)(\eta\otimes
 a\psi) = \widetilde{\rho}(\eta)\otimes \delta_\rho(a)\psi.
\end{equation}
Putting this together with \eqref{eq:86last}, one obtains 
\begin{equation}
\widetilde D_R(\eta a\otimes\psi) - \widetilde D_R(\eta\otimes a\psi)= 0 , \qquad 
\forall a\in\A,\eta\in\E_R,\psi\in\HH.
\label{eq:22}
\end{equation}
Hence the result.
\end{proof}

\noindent 
The explicit form of $\widetilde D_R$, with the Sweedler-like
  notation of \eqref{eq:115}, is 
\begin{equation}
\widetilde D_R(\eta\otimes\psi):= \widetilde{\rho}(\eta) \otimes D\psi +
\widetilde{\rho}(\eta_{(0)}) \ot (\eta_{(1)} \psi), \qquad \qquad \forall \eta\in\E_R,\psi\in\HH.
\label{eq:bis150}
\end{equation}

For the case of a self-Morita equivalence, that is ${\B}= \E_R = \A$,
this operator reduces to a bounded perturbation of $D$ by elements in $\Omega^1_D(\A, \rho)$.
\begin{cor}
\label{corro:right}
  In case $\E_R$ is the algebra $\A$ itself, then $\widetilde D_R =
  D+ \omega_\rho$,  with $\omega_\rho\in \Omega^1_D(\A, \rho)$.
\end{cor}
\begin{proof}
Clearly now $\widetilde{\rho}=\rho$. With $\delta_\rho( \, \cdot \, ) := [D, \, \cdot \, ]_\rho$, as for the non-twisted case recalled in \S \ref{subsec:Moritarightnontwist}, any
  connection $ \nabla$ on $\E_R=\A$ decomposes as 
  \begin{equation}
 \nabla=
   \nabla_0 + {\boldsymbol{\omega}_\rho} \quad \text{where} \quad \left\{
    \begin{array}{cl}
      \nabla_0(a) = \I\otimes \delta_\rho(a) &\text{ is the Grassmann
      connection},\\[6pt]
\boldsymbol{\omega}_\rho(a) = \I\otimes \omega_\rho\,
 a &\text{ with }
     \omega_\rho\in\Omega^1_D(\A, \rho).
    \end{array}\right.
\label{eq:150}
  \end{equation} Hence
  \begin{align}
    \widetilde D_R(a\otimes\psi) &:= \left(\rho\otimes \I\right) \big(a\otimes D\psi +
                               \I\otimes \delta_\rho(a)\psi+ \I \otimes \omega_\rho a\psi \big), \label{eq:47} \\
                             &\:= \rho(a)\otimes D\psi +
                               \I\otimes \delta_\rho(a)\psi+ \I
                               \otimes \omega_\rho a\psi \nonumber \\
\nonumber
                               & \:= \I\otimes(D+ \omega_\rho)a\psi.
      \end{align}
  Identifying $a\otimes\psi=\I\otimes a\psi$ with $a\psi$ and
  $\I\otimes(D+ \omega_\rho)a\psi$ with $(D+ \omega_\rho)a\psi $, one gets that
  $\widetilde D_R$ acts on $\HH \simeq \A \ot_\A \HH$ as $D+\omega_\rho$.
\end{proof}
The operator $D+\omega_\rho$ has a compact
resolvent, being a bounded perturbation of an operator with compact resolvent; 
and  $[D + \omega_\rho, a]_\rho = [D, a]_\rho + [\omega_\rho,
a]_\rho$ is bounded for any $a\in \A$, since $\omega_\rho$ is bounded. Furthermore, 
any grading $\Gamma$ of $(\A, \HH, D)$, will  
anticommutes with $\omega_\rho$, hence with $D+\omega_\rho$. Thus, 
as soon as $\omega_\rho$ is self-adjoint one gets a twisted spectral triple
\begin{equation}
  \label{eq:88}
  (\A, \HH, D+\omega_\rho; \rho) .
\end{equation}
However, and as it happens for the non-twisted case, a priori a real structure $J$ of $(\A, \HH,
D; \rho)$ needs not be a real structure for \eqref{eq:88}. Indeed,
$J(D+\omega_\rho)=
\epsilon' (D + \omega_\rho) J$
if and only if 
$\omega_\rho =J \omega_\rho J^{-1}$
which has no reason to be true due to the following lemma.
\begin{Lemma}
  \label{Lemma:twistopposite}
Let $(\A, \HH, D;\rho)$ together with $J$ be a real twisted spectral triple. With 
\begin{equation}
\omega_\rho=\sum\nolimits_j a_j\,[D, b_j]_{\rho}\,\in\Omega^1_D(\A,
\rho), \label{eq:31}
\end{equation}
one has 
\begin{equation}
 J\omega_\rho J^{-1} = \epsilon'\,\sum\nolimits_j (a_j^*)^\circ[D, (b_j^*)^\circ]_{\rho^\circ}.
\end{equation}
\begin{proof}
Without loss of generality, we may take 
$\omega_\rho=a[D,b]_\rho$. Then
  \begin{align}
  \nonumber
  J \omega_\rho  J^{-1}= J a [D,b]_\rho J^{-1} 
                                                 & = J a J^{-1}
                                                 J[D,b]_\rho J^{-1} =
                                                 {(a^*)}^\circ
                                                 J[D,b]_\rho J^{-1} \nonumber\\ 
                                                 & =
                                                 {(a^*)}^\circ(
                                                 JDbJ^{-1} - J\rho(b)D
                                                 J^{-1})  \nonumber\\
& = \epsilon' (a^*)^\circ (DJb J^{-1} - J\rho(b)J^{-1} D) \nonumber\\
& = \epsilon'  (a^*)^\circ (DJb J^{-1} -\rho^\circ( JbJ^{-1}) D \nonumber\\
& =  \epsilon' (a^*)^\circ [D,(b^*)^\circ]_{\rho^\circ},
    \label{eq:49}
   \end{align}
where we used \eqref{eq:5bis} in the fourth line.
\end{proof}
\end{Lemma}

To implement the self-Morita equivalence of $\A$ in a way which is compatible
with the real structure, one proceeds as in the non-twisted
case, and fluctuates the triple \eqref{eq:88} using also a left module structure thus considering altogether 
an $\A$-bimodule $\E$. 

\subsection{Morita equivalence by left module}
\label{subsec:leftmod}

Let $(\A, \HH, D;\rho), J$ be a real twisted spectral triple. Given a
left $\A$-module $\E_L$,
the right $\A$-module
structure \eqref{eq:35} of $\HH$ allows one to define the Hilbert space
$\HH_L = \HH\otimes_\A \E_L$ 
with Hilbert product recalled in \eqref{eq:176}. As an action of $D$ on $\HH$,
we consider the twist of
the action \eqref{eq:92} by the endomorphism $\widetilde{\rho}^{-1}$, following Assumption \ref{def:end-rho}: 
\begin{align}
(\I\otimes {\widetilde{\rho}}^{-1})\circ D_L: \quad&\HH_L\to \HH_L , \qquad  \psi\otimes\eta \to D\psi\otimes \widetilde{\rho}^{-1}(\eta).
\label{eq:123}
\end{align}
As before, this is not compatible with the tensor product since
\begin{align}
\big( (\I\otimes \widetilde{\rho}^{-1})\!\circ\! D_L \big) (\psi \otimes a\eta)  - (\I\otimes \widetilde{\rho}^{-1})\circ D_L(\psi a\otimes \eta)
&  = D\psi\!\otimes\widetilde{\rho}^{-1}(a\eta) - D(\psi a)\!\otimes \widetilde{\rho}^{-1}(\eta)  \nonumber \\
& = (D\psi) \rho^{-1}(a)\otimes \widetilde{\rho}^{-1}(\eta) - Da^\circ \psi\otimes \!\widetilde{\rho}^{-1}(\eta)   \nonumber \\
&= (\rho^{-1}(a))^\circ D\psi\otimes \widetilde{\rho}^{-1}(\eta) -Da^\circ\psi \otimes\widetilde{\rho}^{-1}(\eta)   \nonumber \\
&= - ([D,a^\circ]_ {\rho^\circ} \psi) \otimes \widetilde{\rho}^{-1}(\eta),
\label{eq:82bis}
\end{align}  
where in the last line we used \eqref{eq:196}.
Again, equation \eqref{eq:82bis} has no reason to vanish. In order to correct it via a connection, one needs to check that
$[D,a^\circ]_{\rho^\circ}$ is actually a derivation. 
%
\begin{Lemma}
The twisted commutator
\begin{equation}
  \delta_\rho^\circ(a): =[D, a^\circ]_{\rho^\circ}\label{eq:157biss}
  \end{equation}
is a derivation of $\A$ in the $\A$-bimodule
\begin{equation}
  \label{eq:89}
  \Omega^1_D(\A^\circ, \rho^\circ):=\Big\{ \sum\nolimits_j a_j^\circ [D,
    b_j^\circ]_{\rho^\circ}, \quad a_j^\circ, b_j^\circ\in \A^\circ\Big\},
\end{equation}
with product law
\begin{equation}
  \label{eq:93}
  a\cdot \omega_\rho^\circ\cdot b := \rho^\circ(b^\circ)\,
  \omega_\rho^\circ \,a^\circ \qquad \forall a,b\in\A,\;
  \omega_\rho^\circ\in\Omega^1_D(\A^\circ, \rho).
\end{equation}
\end{Lemma}
\begin{proof}
By explicit computation of the twisted commutator, one has
  \begin{align}
    \label{eq:100}
   \delta_\rho^\circ(ab) = [D, b^\circ a^\circ]_{\rho^\circ} = \rho^\circ(b^\circ)[D,
    a^\circ]_{\rho^\circ} + [D, b^\circ]_{\rho^\circ}a^\circ
=
    \delta^\circ_\rho(a^\circ)\cdot b + a\cdot\delta_\rho^\circ(b).
  \end{align}
To check that \eqref{eq:89} is a $\A$- bimodule,
first notice that by construction it  is stable under the left
  multiplication by $\A^\circ$, hence under the right
  multiplication by $\A$ defined by \eqref{eq:93}. In addition,
\begin{align}
  \label{eq:86}
\omega^\circ\cdot (ab) &= \rho^\circ((ab)^\circ)\,\omega^\circ =
                         \rho^\circ(b^\circ)\rho^\circ(a^\circ)\omega^\circ =
                         (\omega^\circ\cdot a)\cdot b,
\end{align}
 showing that  $\Omega^1_D(\A^\circ, \rho^\circ)$ is a right
 $\A$-module. Stability for the left multiplication by $\A$ follows from \eqref{eq:100}:
\begin{equation}
  \label{eq:3}
  a\cdot [D,b^\circ]_{\rho^\circ} =  [D,b^\circ]_{\rho^\circ} a^\circ
  = [D,(ab)^\circ]_{\rho^\circ} - [D, a]_\rho^{\circ}\cdot b.
\end{equation}
The left $\A$-module structure is obtained checking that 
  \begin{equation}
    \label{eq:56}
      (ab)\cdot \omega^\circ = w^\circ (ab)^\circ = w^\circ b^\circ a^\circ=
  (b\cdot\omega^\circ )a^\circ=a\cdot(b\cdot\omega^\circ)).
  \end{equation}
Finally, the bimodule structure follows from 
\begin{equation}
  \label{eq:68}
  (a\cdot\omega^\circ)\cdot b= (\omega^\circ_\rho a^\circ)\cdot b =
  \rho^\circ(b^\circ) \,\omega^\circ_\rho \,a^\circ =
  (\omega^\circ_\rho\cdot b)a^\circ = a\cdot(\omega_\rho^\circ\cdot b).
\end{equation}
This finishes the proof.
\end{proof}

Therefore,  the r.h.s. of \eqref{eq:82bis} is --- up to a twist ---
the action on $\HH_L$ of the derivation $\delta_\rho^\circ$. And once again, in
order to define a linear operator on $\HH_L$ using $D$, one needs to correct the action
\eqref{eq:123} with a a connection, this time with value in  $\Omega^1_D(\A^\circ, \rho^\circ)$.
\begin{prop}
\label{prop.leftfluct}
Let $\nabla^\circ$ be an $\Omega^1_D(\A^\circ, \rho^\circ)$-valued connection on
the module $\E_L$. Then the following operator is well defined on $\HH_L$, 
  \begin{equation}
\widetilde D_L := (\I\otimes \widetilde{\rho}^{-1})\circ(D_L + \nabla^\circ) ,
\label{eq:90l}
\end{equation}
where $\nabla^\circ$ denotes the map
induced on $\HH \otimes_\C \E_L$ by the connection, as in \eqref{eq:115b}.
\end{prop}
\begin{proof}
By \eqref{eq:93}, the actions of $\Omega^1_D(\A^\circ, \rho)$ and $\A$ on $\HH_L$ are compatible as in
\eqref{eq:130}, that is, 
\begin{equation}
  \label{eq:101}
  (a\cdot \omega_\rho^\circ)\psi = \omega_\rho^\circ a^\circ \psi =
  \omega^\circ_\rho (\psi a).
\end{equation}
Hence  by
Proposition~\ref{prop:leibnizleft} the connection
$\nabla^\circ$ satisfies the Leibniz rule  
\begin{equation}
\nabla^\circ(\psi \otimes a\eta) - \nabla^\circ(\psi a \otimes \eta)= \delta^\circ_\rho(a)\psi\otimes\eta.
\end{equation}
Therefore
\begin{equation}
  \label{eq:95}
 \left((\I\otimes\widetilde{\rho}^{-1})\circ \nabla^\circ\right)(\psi\ot a\eta)
 -  \left((\I\otimes\widetilde{\rho}^{-1})\circ \widetilde\nabla\right)(\psi a\ot\eta) = \delta^\circ_\rho(a)\psi\ot \rho^{-1}(\eta).
\end{equation}
Together with \eqref{eq:82bis} this yields $\widetilde D_R(\psi\otimes
a\eta) - \widetilde D_R(\psi a\otimes \eta)= 0$, hence the result.
\end{proof}
\noindent
With the Sweedler-like notations of \eqref{eq:115b},  the explicit form of
$\widetilde D_L$ is
\begin{equation}
\widetilde D_L(\psi \otimes \eta):= D\psi \otimes \widetilde{\rho}^{-1}(\eta) +
(\eta_{(-1)}\psi) \ot \widetilde{\rho}^{-1}(\eta_{(0)}).
\label{eq:179}
\end{equation}

To get the more friendly $\widetilde D_L$ for a self-Morita
equivalence, one needs a relation between
$\Omega^1_D(\A, \rho)$ and $\Omega^1_D(\A^\circ, \rho^\circ)$ similar to
the one between  $\Omega^1_D(\A)$ and $\Omega^1_D(\A^\circ)$ given in Lemma \ref{actionoppositeform}.
\begin{Lemma}
\label{lem:actiontwist}
Any $\omega_\rho^\circ= \sum\nolimits_j a_j^\circ [D,
b_j^\circ]_{\rho^\circ}$ in $\Omega^1_D(\A^\circ, \rho^\circ)$ acts on
$\HH$ as 
\begin{equation}
  \label{eq:112}
  \omega^\circ_\rho  =\epsilon'\,  J\omega_\rho J^{-1}
\end{equation}
for
$  \omega_\rho = \sum\nolimits_j a_j^* [D,
b_j^*]_{\rho}\in \Omega^1_D(\A,\rho).$
\end{Lemma}
\begin{proof}
 Without loss of generality, we may take $\omega^\circ_\rho =
 a^\circ[D, b^\circ]_{\rho^\circ}$. Using equation \eqref{eq:5bis} one gets
\begin{align}
  \label{eq:111}
  a^\circ[D, b^\circ]_{\rho^\circ} &= a^\circ Db^\circ - a^\circ\rho^\circ(b^\circ) D =
                               Ja^*J^{-1}DJb^*J^{-1}-J a^*J^{-1}J\rho(b^*)J^{-1}D, \nonumber \\
&= \epsilon '\left( Ja^*Db^*J^{-1} - Ja^*\rho(b^*) D J^{-1}\right)= J\omega_\rho J^{-1}
\end{align}
where $\omega_\rho:= a^*[D, b^*]_\rho\in\Omega^1_D(\A)$.
\end{proof}
\noindent 
In case of a self-Morita equivalence $\B =\E_L=\A$, then $\widetilde D_L$ is just a
bounded perturbation of $D$ by $\Omega^1_D(\A^\circ, \rho)$, similarly
to the right module case of Corollary \ref{corro:right}.
\begin{cor}\label{corro:left}
  In case $\E_L$ is the algebra itself, then
  \begin{equation}
\label{eq:80}\widetilde D_L =
  D + \epsilon' J \omega_\rho  J^{-1}\;\text{ with }\;\omega_\rho\in\Omega^1_D(\A, \rho).
\end{equation}
\end{cor}

\begin{proof} Any $\Omega^1_D(\A^\circ, \rho)$-valued connection
  $\widetilde\nabla^\circ$ on $\E_L=\A$ decomposes
  as  
  \begin{equation}
\nabla^\circ = \nabla^\circ_0 + \boldsymbol{\omega}_\rho^\circ \quad
\text{ where} \quad
\left\{
  \begin{array}{ll}
\nabla^\circ_0(a) = \delta_{\rho}^\circ(a)\otimes \I&\text{ is the Grassmann connection},\\[4pt]
\boldsymbol{\omega}_\rho^\circ(a) = (\omega_\rho^\circ a^\circ) \otimes \I&\text {
where } \omega_\rho^\circ \in \Omega^1_D(\A^\circ, \rho).
  \end{array}\right.
\label{eq:157}
\end{equation} Hence
  \begin{align}
    \label{eq:1310}
    \widetilde D_L(\psi\otimes a)  &= D\psi \otimes \rho^{-1}(a) +\delta_\rho^\circ(a)\psi\otimes \I +
                                 \omega_\rho^\circ a^\circ \psi\otimes \I , \nonumber \\
                               &= (D\psi) \rho^{-1}(a) \otimes\I+ (Da^\circ - \rho^\circ(a^\circ) D)\psi \otimes \I +
                                 \omega_\rho^\circ a^\circ \psi\otimes \I, \nonumber \\                               
                                 &= (\rho^{-1}(a))^\circ D\psi \otimes\I + (Da^\circ - \rho^\circ(a^\circ) D)\psi \otimes \I +
                                 \omega_\rho^\circ a^\circ\psi\otimes \I  \nonumber \\
                                 & = Da^\circ\psi \otimes \I + \omega_\rho^\circ a^\circ \psi\otimes \I,
  \end{align}
where in the last line we used \eqref{eq:196}. By identifying $\HH \otimes_\A \A\simeq \HH$, that is $\psi\otimes a= a^\circ\psi\otimes\I$ with $\psi$ and
similarly for $(Da^\circ\psi) \otimes \I$ and $(\omega_\rho^\circ a^\circ \psi) \otimes \I$, one gets that $\widetilde D_L$ is the
operator $D+\omega_\rho^\circ$. The results follows from Lemma
\ref{lem:actiontwist}, which states that $\omega_\rho^\circ$ acts as
$\epsilon' J\omega_\rho J^{-1}$ for some $\omega_\rho\in\Omega^1_D(\A, \rho)$.
\end{proof}
For reasons similar to those of the right module case, and explained below
Corollary \ref{corro:right}, for a self-adjoint $\omega_\rho$ one has
that 
the triple $(\A, \HH, D+ \epsilon' J \omega_\rho J^{-1})$ is a (graded) twisted spectral triple, failing
to necessarily admit $J$ as a real structure, thus the need of a bimodule.  
\begin{rem}
\textup  
{In \eqref{eq:123}, we have used $\rho^{-1}$ rather than $\rho$, so 
that the failure of linearity is captured by
    $\delta_\rho^\circ$. Twisting by $\rho$, one would arrive at
    $\delta_{\rho^{-1}}^\circ$. Alternatively one may require that
    $\rho$ is a
    $*$-automorphism: equation \eqref{eq:170} then implies $\rho^{-1}=\rho$.
    }
\end{rem}
\subsection{Bimodule and the real structure}
\label{subsec:bimodandreal}
To make the real structure compatible with Morita equivalence of
twisted spectral triples, one combines the  two
constructions above in a way similar to the non-twisted case.  
Firstly fluctuate the real twisted spectral triple $(\A, \HH, D; \rho), J$ using the right module
$\E_R=\A$, then fluctuate the resulting triple \eqref{eq:88} via the left
module $\E_L=\A$. This yields the triple $ (\A, \HH, D' )$
where 
\begin{equation}
  \label{eq:33bis}
 D' =D + \omega_\rho^L + \epsilon' J\omega^R_\rho J^{-1}
\end{equation}
with $\omega_\rho^R$ and $\omega_\rho^L$  two elements of $\Omega^1_D(\A,\rho)$ 
that are a priori distinct. 
\begin{prop}
\label{prop:twistMorita}
It holds that $D'J=\epsilon' D'J$ if and only if there exists an element 
$\omega_\rho$ in $\Omega^1_D(\A, \rho)$ such that 
 \begin{equation}
   D' = D + \omega_\rho + \epsilon'J\omega_\rho J^{-1}.
 \end{equation}
\end{prop}
\begin{proof}
From \eqref{eq:33bis}, one finds that $JD'= \epsilon'D'J$ if and only if
\begin{equation}
  \label{eq:46bis}
 (\omega_\rho^R-\omega_\rho^L) - \epsilon' J (\omega_\rho^R  -\omega_\rho^L)J^{-1} =0.
\end{equation}
Adding half of this expression to the r.h.s. of \eqref{eq:33bis}, one gets
\begin{equation}
  \label{eq:64bis}
  D' = D + \tfrac 12(\omega_\rho^R +\omega_\rho^L) + \epsilon' J \tfrac 12(\omega_\rho^R + \omega_\rho^L) J^{-1}.
\end{equation}
Hence the result with $\omega_\rho:= \tfrac 12(\omega_\rho^R + \omega_\rho^L)$.
\end{proof}

Proposition \ref{prop:twistMorita}  shows that Morita equivalence together with the real
structure yields the twisted-fluctuation \eqref{eq:76}. This answers
the first question raised in the introduction, and puts the
twisted-gauged Dirac operator $D' = D_{\omega_\rho}$ on the same
footing as the covariant operator $D_\omega$, namely as a covariant derivative associated
to a connection.  The only dif and only if erence is that, 
in the twisted case, the action $D_{R,L}$
of the Dirac operator on $\HH_{R,L}$ and the
action of the $\Omega^1_D(\A, \rho)$-valued connection have to be
twisted by $(\I\otimes\widetilde{\rho})$ and $(\I\otimes\widetilde{\rho}^{-1})$. 

\begin{rem}
\textup{
  \label{remark:twistwork}
It is worth stressing that fluctuations by Morita equivalence translate to the twisted case because the
conditions \eqref{eq:129} and  \eqref{eq:130}, that allow one to pass the
Leibniz rule from the connection, as a map on $\E$, to the connection as a map on 
$\E_R\otimes \HH$ or $\HH\otimes\E_L$, are still valid in the
twisted case, that is it holds that
\begin{align}
  \label{eq:119}
  (\omega_\rho\cdot a)\psi = \omega_\rho a\psi= \omega_\rho (a\psi), \qquad
\psi (a\cdot \omega^\circ_\rho) = \omega^\circ_\rho\, a^\circ \psi =
  (\psi a) \omega^\circ_\rho. 
\end{align}
}
\end{rem}
\begin{rem} 
\textup{
By choosing the Grassmann connection, that is $\omega=0$ in Corollaries
  \ref{corro:right} and \ref{corro:left}, one gets $\widetilde D_L
  =\widetilde D_R =D$, so that $D'=D$ in \eqref{eq:33bis}. In other terms --- and as in the non-twisted case
  (see Remark \ref{fluctgrassmann}) --- implementing the self-Morita
  equivalence of $\A$ in a twisted spectral triple with the Grassmann connection yields no 
  fluctuation $D_{\omega_\rho}=D$.   
}
\end{rem}

\section{Twisted gauge transformation}
\label{sec:twisgaugetransform}

A \emph{gauge transformation} on a module $\E$ is the action of
a unitary endomorphism $u$ of $\E$ on a
$\Omega$-valued connection $\nabla$ on the module (see  \S\ref{sec:annexgauge} for details),
\begin{equation}
 \nabla \to \nabla^u:=u\nabla u^*\quad u\in{\cal U}(\E).
\label{eq:137}
 \end{equation} 
Given a spectral triple $(\A, \HH, D)$, with $\E=\A$ and
$\Omega=\Omega^1_D(\A)$, a gauge transformation 
by $u=u^*$ a unitary element of $\A$,
amounts to substituting the
gauged Dirac operator $D+ \omega + J\,\omega\, J^{-1}$ with
$D+ \omega^u + J\,\omega^u\,J^{-1}$
where
\begin{equation}
  \label{eq:147}
  \omega^u = u[D, u^*] + u\,\omega\, u^*,
\end{equation}
(see \eqref{eq:143}).
This transformation maps a self-adjoint
$\omega\in\Omega^1_D(\A)$ to a self-adjoint $\omega^u\in\Omega^1_D(\A)$,
and gives the usual
transformation rule of the gauge potential when applied to 
almost commutative geometry (that is the product of a manifold
by a finite dimensional spectral triple). 

It is clear that \eqref{eq:147}
cannot be valid in the twisted case, when one considers a
connection with value in the bimodule of twisted $1$-forms. Indeed, given $\omega_\rho\in\Omega^1_D(\A,
\rho)$,  there is no reason for $u[D, u^*] + u\,\omega_\rho\, u^*$ to be in $\Omega^1_D(\A, \rho)$,
because $[D, u^*]$ has no reason to be in  $\Omega^1_D(\A,
\rho)$ (let alone to be a bounded operator). We show in \S \ref{subsec:potential}
that a gauge
transformation \eqref{eq:137} in fact substitutes $\omega_\rho$  in the twisted-gauged Dirac operator
$D_{\omega_\rho} = D +\omega_\rho + \epsilon' J \omega_\rho J^{-1}$ with 
\begin{equation}
  \label{eq:148}
   \omega_\rho^u := \rho(u)[D, u^*]_\rho + \rho(u)\omega_\rho u^*. 
\end{equation} 
Furthermore,  we show in \S \ref{subsec:twistedunitaries} that a gauge
transformation is equivalent to 
the twisted conjugate action on the Dirac operator of the adjoint
representation \eqref{eq:106} of the unitaries of $\A$, that is,
\begin{equation}
D_{\omega_\rho^u}= \rho(U) D_{\omega_\rho} U^* \quad \text{ for }\; U=\Ad(u), \quad u\in\U(\A).\label{eq:149}
\end{equation}

\subsection{Transformation of the gauge potential}
\label{subsec:potential}

In all this section, $(\A, \HH, D;\rho), J$ is a real twisted spectral
triple, $\E$ a hermitian $\A$-module and  ${\cal U}(\E)$ its group of unitary
endomorphisms.
\begin{Lemma} 
\label{Lemmatwistaction}
Let $\nabla$ be a $\Omega^1_D(\A, \rho)$-valued connection on
  $\E$. Then, for any $u\in{\cal U}(\E)$ one has
\begin{align}
  \label{eq:156}
(\widetilde{\rho}\otimes \I) \, \nabla^u =  (\widetilde{\rho}\circ u)\nabla u^* & \qquad \textup { for a
                                                      right module},\\
\label{eq:156bis}
(\I\otimes \widetilde{\rho}) \, \nabla^u =  (\widetilde{\rho}\circ u)\nabla u^* & \qquad \textup { for a
                                                      left module},
\end{align}
with $\nabla^u$ the gauge transformation \eqref{eq:137} and
$\widetilde{\rho}$ the endomorphism of $\E$ in the Assumption~\ref{def:end-rho}. In
particular, taking for $\E$ the algebra itself, one gets 
\begin{align*}
(\rho\otimes \I)\nabla^u(a) &= \rho(u)\otimes \delta_\rho(u^*a) + \rho(u)\otimes  \omega_\rho u^*a 
&\quad \textup{for $\E=\E_R=\A$ as right module},\\
(\I\otimes \rho)\nabla^u(a) &=\delta_\rho^\circ(au) \otimes \rho(u^*) +  \omega^\circ_\rho (au)^\circ \otimes \rho(u^*)
 &\textup{for $\E=\E_L=\A$ as left module},
\end{align*}
where now $u$ is a unitary element of $\A$, while
$\omega_\rho\in\Omega^1_D(\A, \rho)$ and
$\omega_\rho^\circ\in\Omega^1_D(\A^\circ, \rho^\circ)$ are the $1$-forms associated to $\nabla$ as
defined in Corollaries \ref{corro:right} and \ref{corro:left}.
\end{Lemma}
\begin{proof}
Assume $\E$ is a right $\A$-module. For any $\eta\in\E$ and $u\in{\cal
  U}(\E)$, write
$\nabla(u^*(\eta))=\eta^u_{(0)}\otimes\eta^u_{(1)}$ with $\eta^u_{(0)}\in\E$ and 
 $\eta^u_{(1)}\in\Omega^1_D(\A, \rho)$ (with an implicit sum). By \eqref{eq:103} one gets on
 the one hand
 \begin{equation}
    \label{eq:151}
 \left((\widetilde{\rho}\otimes \I) (u\nabla u^*)\right)(\eta)= (\widetilde{\rho}\otimes \I)
 \left(u (\eta^u_{(0)})\otimes \eta^u_{(1)}\right) =
 \widetilde{\rho}(u(\eta^u_{(0)}))\otimes \eta^u_{(1)},\end{equation}
while  on the other hand
\begin{equation}
 \left((\widetilde{\rho}\circ u)\nabla u^*\right)(\eta) = (\widetilde{\rho}\circ u) (\eta^u_{(0)}\otimes \eta^u_{(1)})= 
 \left(\widetilde{\rho}\circ u(\eta^u_{(0)})\right)\otimes\eta^u_{(1)} =
 \widetilde{\rho}(u(\eta^u_{(0)}))\otimes \eta^u_{(1)}.
\end{equation}
Hence \eqref{eq:156}. The proof is similar for a left
$\A$-module.  

For the second part of the lemma, for any $a\in\E_R\simeq\A$ with $\nabla=\nabla_0 + {\bo_\rho}$, by \eqref{eq:128} and \eqref{eq:150} one writes the r.h.s. of \eqref{eq:156} as 
\begin{align}
  \label{eq:1540}
((\rho\circ u)\nabla u^*)(a) &=((\rho\circ u)\nabla)(u^*a) =
                               (\rho\circ u)(\nabla_0 (u^*a)+
                               {\bo_\rho}(u^*a)), \nonumber \\
&=(\rho\circ u)( \I\otimes\delta_\rho(u^*a)  + \I \otimes \omega_\rho u^*a) 
\nonumber \\
& = \rho(u)\otimes \delta_\rho(u^*a) + \rho(u)\otimes  \omega_\rho u^*a.
\end{align}
Similarly, for $a\in\E_L \simeq \A$ with
$\nabla=\nabla_0^\circ +{\bo_0^\circ}$, by \eqref{eq:128}  and
\eqref{eq:157},  the r.h.s. of \eqref{eq:156bis} reads
\begin{align}
  \label{eq:15400}
((\rho\circ u)\nabla u^*)(a) &=((\rho\circ u)\nabla)(au) =
                               (\rho\circ u)(\nabla^\circ_0 (au)+
                               {\bo^\circ_\rho}(au)), \nonumber \\
\nonumber &=(\rho\circ u)( \delta_\rho^\circ(au) \otimes \I+\omega^\circ_\rho
  (au)^\circ\otimes \I) \nonumber \\ &= \delta_\rho^\circ(au) \otimes \rho(u^*) +
  \omega^\circ_\rho (au)^\circ \otimes \rho(u^*).
\end{align}
Hence the result.
\end{proof}
\noindent  
A gauge transformation \eqref{eq:137} amounts to substituting
$(\widetilde{\rho}\ot\I)\circ\nabla$ with $(\widetilde{\rho}\ot\I)\circ \nabla^u$ in the
definition \eqref{eq:90} of $\widetilde D_R$, and $(\I\ot\widetilde{\rho}^{-1})\circ\nabla^\circ$ with 
$(\I\ot\widetilde{\rho}^{-1})\circ\nabla^{\circ u}$
in the definition \eqref{eq:90l}  of $\widetilde D_L$.
For the cases $\E=\A$, one obtains the following explicit formulas.

\begin{prop} For a gauge transformation with a unitary $u\in\A$, 
the operators
  $\widetilde D_R= D + \omega_\rho$ and $\widetilde D_L = D + \omega_\rho^\circ$ of
  Corollaries \ref{corro:right} and \ref{corro:left} are mapped to $\widetilde D_R^u =D + \omega_\rho^u$
  and  $\widetilde D_L^u =D + {\omega_\rho^\circ}^u$
where the transformed twisted $1$-forms are given by
\begin{align}
    \label{eq:156ter}
\omega_\rho^u &:= \rho(u)[D, u^*]_\rho + \rho(u)\,\omega_\rho\, u^*\\
\label{eq:156quat}
{\omega_\rho^\circ}^u&:=   \rho^\circ(u^{*\circ}) [D,{u}^\circ]_{\rho^\circ} + \rho^\circ(u^{*\circ})\, \omega_\rho^\circ\,{u}^\circ.
  \end{align}
  \end{prop}
\begin{proof}
 By Lemma \ref{Lemmatwistaction}, substituting $\nabla$ with
 $\nabla^u$ in \eqref{eq:47} yields the operator 
\begin{align}
  \widetilde D_R^u(a\otimes \psi) &=\rho(a)\otimes D\psi + \rho(u)\otimes\delta_\rho(u^*a)\psi + \rho(u)\otimes  \omega_\rho u^*a\psi, \nonumber \\
&=\I\otimes \left(\rho(a)D + \rho(u) [D, u^*a]_\rho\right)\psi +\I\otimes \rho(u) \omega_\rho u^*a\psi, \nonumber \\
&=\I\otimes \left(D + \rho(u)[D,u^*]_\rho\right)a\psi + \I\otimes \rho(u) \omega_\rho u^*a\psi,
\label{eq:1550}
\end{align}
where in the last line we used
\begin{equation}
\rho(a)D + \rho(u) [D, u^*a]_\rho = \rho(u)Du^*a = (D +\rho(u)[D, u^*]_\rho)a.
\end{equation}
Identifying $a\otimes \psi=\I\otimes
a\psi$ with $a\psi$ in $\A \otimes_\A \HH \simeq \HH$, equation \eqref{eq:1550} shows that $\widetilde D_R^u$
acts on $\HH$ as $D+\omega_\rho^u$ with $\omega_\rho^u$ as in \eqref{eq:156ter}. 

\noindent
Similarly, substituting $\nabla^\circ$ with
$\nabla^{\circ u}$ 
in  \eqref{eq:1310} yields the operator
\begin{align}
  \widetilde D_L^u(\psi\otimes a) &= D\psi \otimes \rho^{-1}(a)+ \delta^\circ(au)\psi \otimes\rho^{-1}(u^*)+ \omega_\rho^\circ (au)^\circ\psi\otimes \rho^{-1}(u^*), \nonumber \\
&= \left((\rho^{-1}(a))^\circ D  +
(\rho^{-1}(u^*))^\circ[D,(au)^\circ]_{\rho^\circ}\right)\psi\otimes \I 
+ (\rho^{-1}(u^*))^\circ \omega_\rho^\circ (au)^\circ\psi\otimes \I,
  \nonumber \\
&= \left(\rho^\circ(a^\circ) D  +\rho^\circ(u^{*\circ})[D,(au)^\circ]_{\rho^\circ}\right)\psi\otimes \I 
+ \rho^\circ(u^{*\circ}) \omega_\rho^\circ (au)^\circ\psi\otimes \I, \nonumber\\
&= \left(D+\rho^\circ(u^{*^\circ}) [D,{u}^\circ]_{\rho^\circ}\right)a^\circ \psi\otimes\I + \rho^\circ(u^{*\circ}) \omega_\rho^\circ {u}^\circ a^\circ\psi\otimes \I, 
\label{eq:155ter}
\end{align}
where we used \eqref{eq:196} and, in the last line,
\begin{equation}
  \rho^\circ(a^\circ) D  +\rho^\circ(u^{*\circ})[D,(au)^\circ]_{\rho^\circ} = \rho^\circ(u^{*\circ}) D (au)^\circ = (D + \rho^\circ(u^{*\circ}) [D, {u}^\circ]_{\rho^\circ}) a^\circ.
\end{equation}
Identifying $\psi\otimes a = a^\circ\psi\otimes\I$ with $a^\circ\psi$ in $\HH \otimes_\A \A \simeq \HH$
equation \eqref{eq:155ter} shows that $\widetilde D_L^u$ acts on $\HH$ as
$D+{\omega_\rho^\circ}^u$, with ${\omega_\rho^\circ}^u$ as defined in \eqref{eq:156quat}.
 \end{proof}
 
\begin{prop}
\label{prop:twistedgaugefluct}
 In a twisted spectral triple $(\A, \HH, D;\rho)$, the law of transformation of the gauge
 potential is $\omega_\rho \to \omega_\rho^u$, with $\omega_\rho^u$ given in \eqref{eq:148}. 
\end{prop}
\begin{proof}
  By Lemma \ref{lem:actiontwist}, we substitute  $\omega_\rho^\circ$ in
  \eqref{eq:156quat} with $\epsilon' J\omega_\rho^LJ^{-1}$,  with
  $\omega_\rho^L\in\Omega^1_D(\A, \rho)$. Explicitly, using
  \eqref{eq:196} and \eqref{eq:170} to write
  \begin{equation}
\rho^\circ(u^{*\circ})=(\rho^{-1}(u^*))^\circ=(\rho(u)^*)^\circ= J\rho(u)J^{-1},
  \label{eq:146}
  \end{equation}
one obtains
  \begin{align}
  \label{eq:159}
  {\omega_\rho^\circ}^u &= \rho^\circ(u^{*\circ}) [D,{u}^\circ]_{\rho^\circ} + \rho^\circ(u^{*\circ})
    \omega_\rho^\circ{u}^\circ \nn \\
&  = J\rho(u)J^{-1} [D,{u}^\circ]_{\rho^\circ} +J\rho(u)J^{-1}\omega_\rho^\circ{u}^\circ \nn \\
 &= \epsilon' J\left(\rho(u) [D, u^*]_{\rho} + \rho(u)\,\omega_\rho^L\,
                              u^*\right)J^{-1} =\epsilon' J\,(\omega_\rho^L)^{u}\, J^{-1},
\end{align}
where in the third line we used again \eqref{eq:146} to write
\begin{equation}
  \label{eq:180}
 [D, u^\circ]_{\rho^\circ}= Du^\circ - \rho^\circ(u^\circ) D=
 DJu^*J^{-1} - J\rho(u^*)J^{-1}D=\epsilon'J[D,u^*]_\rho J^{-1}.
\end{equation}
Therefore, with the notation of Proposition \ref{prop:twistMorita},
 one has that $\omega_\rho= \tfrac 12(\omega_\rho ^R +\omega_\rho^L)$ is
 mapped under a gauge transformation to 
 \begin{align}
   \label{eq:163}
   \tfrac 12( (\omega_\rho ^R)^u +(\omega_\rho^L)^u ) & = \rho(u)[D,
   u^*]_\rho + \tfrac 12 \rho(u)(\omega_\rho^R +\omega_\rho^L) u^*, \nonumber \\
&= \rho(u)[D,
   u^*]_\rho + \rho(u) \omega_\rho u^*,
 \end{align}
that is $\omega_\rho^u$ as defined in \eqref{eq:148}.
\end{proof}
\noindent The transformation of the gauge
potential of a  twisted
spectral triple is thus the
usual gauge transformation \eqref{eq:143}, in which the left action of
$u$ and
the commutator have been twisted by the automorphism. This suggests that a twisted
fluctuation may also be obtained by twisting the left action of
$\Ad(u)$ in \eqref{eq:11}.
 
\subsection{Twisted action of unitaries}
\label{subsec:twistedunitaries}

Let $\Ad(u)=uJuJ^{-1}$ denote the adjoint action on $\HH$ of a unitary $u\in\A$ as recalled in 
\eqref{eq:106}:
\begin{equation}
  \Ad(u)\psi := u \,\psi \,u^*  = u J  u J^{-1}\psi  \qquad \forall \psi\in\HH.
\end{equation}
We show in Proposition \ref{prop:twistedunitaries} that a twisted conjugation
by $\Ad(u)$ of a twisted-gauged Dirac operator
$D_{\omega_\rho}$ yields the gauge transformation of Proposition \ref{prop:twistedgaugefluct}. 
Preliminarily, we begin by proving a twisted version of \eqref{eq:11}. 
\begin{Lemma}
\label{lem:twistedunit}
Let $(\A, \HH, D; \rho)$ be a real twisted spectral triple; for any $u\in\U(\A)$ define
\begin{equation}
\label{eq:136a}
  \rho(\Ad(u)):=  \Ad(\rho(u))= \rho(u) J\rho(u) J^{-1}.
\end{equation}
Then, it holds that 
\begin{equation}
  \label{eq:19}
      \rho(\Ad(u))\, D \Ad(u)^{-1} = D  +\rho(u) [D, u^*]_\rho +
\epsilon'  J \,\rho(u)[D, u^*]_\rho \, J^{-1}.
\end{equation}
\end{Lemma}
\begin{proof}
  Let $v:= JuJ^{-1}\in\A^\circ$. By \eqref{eq:146} one has $\rho^\circ(v) =\rho^\circ(u^{*\circ})=
  J \rho(u)J^{-1}$, so that
  \begin{equation}
  \Ad(u)=uv=vu,\qquad \rho(\Ad(u))=\rho(u)\rho^\circ(v) = \rho^\circ(v) \rho(u)
\label{eq:153}
  \end{equation}
by the order $0$ condition.
Using the twisted first-order condition \eqref{eq:13} one computes:
  \begin{align}
    \nonumber   \rho(\Ad(u))\, D \Ad(u)^{-1}&= \rho^\circ(v)\,(\rho(u)Du^{-1})\, v^{-1}
  = \rho^\circ(v)(D +\rho(u) [D, u^{-1}]_\rho) v^{-1}\\
    \nonumber &= \rho^\circ(v)D  v^{-1} + \rho^\circ(v)\rho(u)[D, u^{-1}]_\rho v^{-1}\\
  \label{eq:1900}  &= D +\rho^\circ(v) [D, v^{-1}]_{\rho^\circ}  + \rho(u) [D, u^{-1}]_\rho.
  \end{align}
By \eqref{eq:4}, one has
  \begin{align}
    \nonumber \rho^\circ(v)[D, v^{-1}]_{\rho^\circ} &=  \rho^\circ(v) D
    v^{-1}  - D =   \epsilon' J \rho(u) D u^{-1} J^{-1}-D,\\
    \nonumber&=   \epsilon' J \left(D +\rho(u)[D, u^{-1}]_\rho\right) J^{-1}-D,\\
&=  (\epsilon')^2 D + J\epsilon' \rho(u)[D, u^{-1}]_\rho J^{-1}-
    D = \epsilon' J \rho(u)[D, u^{-1}]_\rho J^{-1}.
  \end{align}
Plugged into \eqref{eq:1900}, one gets \eqref{eq:19}.
\end{proof}

\begin{prop}
\label{prop:twistedunitaries}
Let $(\A, \HH, D; \rho), J$ be a real twisted spectral triple and consider a twisted
gauged Dirac operator $D_{\omega_\rho}= D + \omega_\rho + \epsilon' J \omega_\rho J^{-1}$ as in \eqref{eq:3}. Then for any $u\in\U(\A)$ one has
    \begin{equation}
    \rho(\Ad(u)) \,D_{\omega_\rho} \,\Ad(u)^{-1} = D +
    \omega_\rho^u + \epsilon'\,J\, \omega_\rho^u \,J^{-1}\label{eq:9}
  \end{equation}
  with transformed $\omega_\rho^u$ given in \eqref{eq:148}.
\end{prop}
\begin{proof} 
For
  $\omega_\rho = a  [D, b]_\rho$ (without loss of generality), one needs to compute $ \rho(\Ad(u))\,\omega_\rho\Ad(u)^{-1}$ and
  $\rho( \Ad(u))\, J\omega_\rho J^{-1} \Ad(u)^{-1}$. By the twisted first-order condition one
  gets
  \begin{align}
     \rho(\Ad(u)) \, \omega_\rho \Ad(u)^{-1}&= \rho(u)\,\left(\rho^\circ(v) a  [D, b]_\rho v^{-1}\right)\, u^{-1},
     \nonumber \\
&= \rho(u)\left(a  [D, b]_\rho \right)u^{-1}= \rho(u) \omega_\rho u^{-1}.
\label{transjauge}
  \end{align}
 In order  to compute $\rho(\Ad(u))\,J\omega_\rho J^{-1} \Ad(u)^{-1}$, one uses on the one hand,
\begin{align}
\rho(u) \,J \omega_\rho J^{-1} \,u^{-1}&= \rho(u)\, J a  [D, b]_\rho J^{-1}\, u^{-1}
= J\,( J^{-1}\rho(u) J)\, a  [D, b]_\rho J^{-1} \,u^{-1}, \nonumber \\
&= J a  [D, b]_\rho\,( J^{-1}u J)\, J^{-1} u^{-1} =
 J a  [D, b]_\rho J^{-1} =  J \omega_\rho  J^{-1},
\end{align}
and on the other hand
\begin{align}
\rho^\circ(v) \,J \omega_\rho J^{-1} \,v^{-1} = J\rho(u) J^{-1}\, J \omega_\rho J^{-1}\, J u^{-1} J^{-1}
= J\rho(u) \omega_\rho u^{-1} J^{-1},
\end{align}
so that
\begin{equation}
\rho(\Ad(u)) J \omega_\rho J^{-1} \Ad(u)^{-1}= \rho^\circ(v) \,\left(\rho(u)\, J \omega_\rho J^{-1}\, u^{-1}\right)\, v^{-1}
  \label{eq:21}= J \rho(u) \omega_\rho u^* J^{-1}.
 \end{equation}
Collecting \eqref{eq:21} and \eqref{transjauge} one gets
\begin{equation}
  \label{eq:23}
  \rho(\Ad(u)) \left( \omega_\rho + \epsilon' J \omega_\rho J^{-1}\right)
   \Ad(u)^{-1} = \rho(u) \omega_\rho u^{-1} + \epsilon' J \rho(u) \omega_\rho u^{-1} J^{-1}.
\end{equation}
Together with \eqref{eq:19}, this  yields the result.
\end{proof}

\section{Self-adjointness}
\label{sec:selfadjoint}
 
In the non twisted case, a gauge transformation preserves the
self-adjointness of the Dirac operator. The transformed operator 
\begin{equation}
  D_{\omega^u}= \Ad(u)) \, D_\omega \Ad(u)^{-1}  
\end{equation}
 is self-adjoint if and
only if so is $D_\omega$, since $\Ad(u)$ is
unitary (see Lemma \ref{lem:unitaryad} below).  Thus, starting with
a spectral triple $(\A, \HH,
D_{\omega})$, a gauge transformation yields a spectral triple
$(\A, \HH, D_{\omega^u})$, which is unitary equivalent to
the former \cite{Connes:1996fu}. This result is no longer true in the twisted case:
by Proposition~\ref{prop:twistedunitaries} the gauge transformed of the
twisted-gauged Dirac operator $D_{\omega_\rho}$ is 
\begin{equation}
  \label{eq:7}
  D_{\omega_\rho^u}=\rho(\Ad(u)) \, D_{\omega_\rho} \Ad(u)^*,
\end{equation}
which has no reason to be self-adjoint, even if $D_{\omega_\rho}$ is
self-adjoint.

We next work out conditions on the  unitary element $u$ to
guarantee that the operator $D_{\omega_\rho^u}$ be self-adjoint.
A simple condition would be that $u$ is invariant for the the twist: $\rho(u)=u$.
We show, for the example of the minimal twist of a spin manifold constructed in \cite{Landi:2015aa},   
that there exists other solutions than this trivial one. 

\subsection{Conditions for self-adjointness}

Let us begin with recalling some properties
of antilinear operators. The adjoint of an antilinear operator $C$ on
a Hilbert space $\HH$ is the
antilinear operator $C^*$ such that 
\begin{equation}
  \label{eq:10}
  \langle C\xi, \zeta \rangle = \overline{\langle \xi,
    C^*\zeta\rangle}, \quad \forall \xi,\zeta\in\HH.
\end{equation}
Such an operator is antiunitary if
\begin{equation}
  \label{eq:152}
  \langle C\xi, C\zeta\rangle = \overline{\langle \xi, \zeta\rangle},
 \quad \text{ that is } \quad C^*C= CC^*=\I.
\end{equation}
Hence $C^* = C^{-1}$, as for linear unitary
 operators. However, one has to be careful that
  the
 usual rule for the adjoint holds for the product of two antilinear operators $C, C'$, 
 \begin{equation}
   \label{eq:166}
 \langle CC'\xi, \zeta\rangle = \overline{\langle C'\xi,
   C^*\zeta\rangle} =  \langle \xi, C'^*C\zeta\rangle \quad \text{so that}\quad    (C C')^* = C'^* C^*
 \end{equation}
but not for the product of an antilinear $C$ with a linear $T$, for
\begin{equation}
  \label{eq:165}
   \langle CT\xi, \zeta\rangle = \overline{\langle T\xi,
   C^*\zeta\rangle} = \overline{\langle \xi, T^*C\zeta\rangle}.
\end{equation}
On the other hand, the usual rule for the adjoint holds for any product involving an even number 
of antilinear operators, such as $JTJ^{-1}$  with $T$ linear, that 
often appear in this work. This is shown explicitly in the following lemma for $T=u$ an unitary element. 

\begin{Lemma}
\label{lem:unitaryad}
Let $(\A, \HH, D)$ be a real spectral triple with real structure (the
antilinear operator)~$J$. Then, for any unitary $u\in\A$ one has that  $\Ad(u) = uJuJ^{-1}$ is unitary.
 \end{Lemma}
\begin{proof}
The operator $JuJ^{-1}$ is linear, hence $\Ad(u)^* =
  (JuJ^{-1})^*u^*$. A direct computation yields
    \begin{equation}
\label{eq:antil}
  \langle JuJ^{-1}\xi, \zeta \rangle =   \overline{\langle uJ^{-1}\xi,
    J^*\zeta \rangle}= \overline{\langle J^{-1}\xi,
   u^* J^*\zeta \rangle}=\langle\xi,
   (J^{-1})^* u^* J^*\zeta \rangle
\end{equation}
that is, using $J^*=J^{-1}$, 
\begin{equation}
  \label{eq:168}
  (JuJ^{-1})^* = J u^* J^{-1}.
\end{equation}
Hence $\Ad(u)^*= Ju^*J^{-1}u^*$, so that $\Ad(u)^*\Ad(u)=\Ad(u)\Ad(u)^*=\I$.
 \end{proof}

We now work out a condition on a unitary element $u$ which is  equivalent to $D_{\omega_\rho^u}$ 
being self-adjoint. Taking advantage of the two formulas for $D_{\omega_\rho^u}$ (the two sides
of \eqref{eq:19}), we actually exhibit two conditions which are equivalent.
\begin{prop}
\label{prop:conditionsselfadjoint}
  Let $(\A, \HH, D; \rho), J$ be a real twisted spectral triple, $D_{\omega_\rho}$ a twisted-gauged Dirac operator 
  and $u$ a unitary element of $\A$. Then the gauge transformed operator $D_{\omega_\rho^u}$ in \eqref{eq:7} is
    self-adjoint if and only if
\begin{equation}
  \label{eq:178bis}
 J\omega(u)J^{-1}= -\epsilon'\omega(u),
\end{equation}
for 
\begin{equation}
  \label{eq:177}
   \omega(u)= u^\circ \,[D, \rho(u)^*u]_\rho \,{u^*}^\circ \quad \text{ or
   }\quad  \omega(u)= u\, [D, \rho(u)^*u]_\rho u^* ,
\end{equation}
the two choices being equivalent.
\end{prop}
\begin{proof}
We write $D=D_{\omega_\rho}$, taken to be self-adjoint. Then 
$D_{\omega_\rho^u}=\rho(\Ad(u))D\Ad(u)^*$ is self-adjoint by Lemma
  \ref{lem:unitaryad} if and only if 
  $\rho(\Ad(u))D\Ad(u)^* = \Ad(u) D\rho(\Ad(u))^*$ or, equivalently
  \begin{equation}
\Ad(u)^*\rho(\Ad(u))D = D\rho(\Ad(u))^*\Ad(u).
\label{eq:99}
\end{equation}
By \eqref{eq:170} and \eqref{eq:12},
\begin{equation}
  \label{eq:167}
  \rho(\Ad(u))^*= \rho(u)^*J\rho(u)^*J^{-1}=\rho^{-1}(u^*) \,J\,
  \rho^{-1}(u^*) J^{-1}=\rho^{-1}(\Ad(u)^*),
\end{equation}
so that $\rho\left(\rho(\Ad(u))^*\Ad(u)\right)=
\Ad(u)^*\rho(\Ad(u))$. Hence condition \eqref{eq:99}
becomes
\begin{equation}
  \label{eq:172}
  [D, \rho(\Ad(u))^*\Ad(u)]_\rho=0.
\end{equation}
By the order zero condition, one has
\begin{align}
\rho(\Ad(u))^*\Ad(u) = \rho(u)^*J\rho(u)^*J^{-1}\,uJuJ^{-1} = \rho(u)^*u
\,J\rho(u)^*uJ^{-1}= \mathfrak{u} J\mathfrak{u} J^{-1}
\label{eq:173}
\end{align}
where $\mathfrak{u} := \rho(u)^*u$. Therefore
\begin{align}
    [D, \rho(\Ad(u))^*\Ad(u)]_\rho&=  [D, \mathfrak{u} J\mathfrak{u}
  J^{-1}]_\rho \nn \\
&=\rho(\mathfrak{u}) [D, J\mathfrak{u} J^{-1}]_\rho + [D, \mathfrak{u}]_\rho J\mathfrak{u} J^{-1}
  \label{eq:174bis}
=\epsilon' J\omega(\mathfrak{u}) J^{-1}  +\omega(\mathfrak{u}),
\end{align}
with
\begin{equation}
  \label{eq:181}
  \omega(\mathfrak{u}) := J\rho(\mathfrak{u}) J^{-1}[D, \mathfrak{u}]_\rho,
\end{equation}
where we used the twisted first-order condition  as well as
\begin{align*}
\nonumber
  \rho(\mathfrak{u}) [D, J\mathfrak{u} J^{-1}]_\rho &= J^{-1}\left(J\rho(\mathfrak{u})J^{-1}\, J [D, J\mathfrak{u}  J^{-1}]_\rho\right) \\ &= \epsilon' J^{-1}\left(J\rho(\mathfrak{u})J^{-1}\,\left( D J^2\mathfrak{u}-  J^2\mathfrak \rho(\mathfrak{u}) D\right)\right) J^{-1},\\
&= \epsilon' \epsilon'' J^{-1}\left(J\rho(\mathfrak{u})J^{-1}[D,\mathfrak
  u]_\rho\right) J^{-1} 
= \epsilon' J\left(J\rho(\mathfrak
  u)J^{-1}[D,\mathfrak{u}]_\rho\right) J^{-1}=\epsilon' J \omega(\mathfrak{u}) J^{-1}.
\end{align*}
The first part of the proposition follows from \eqref{eq:174bis}, noticing that
\begin{equation}
\rho(\mathfrak{u})= \rho(\rho(u)^*u) =
\rho(\rho(u)^*)\rho(u) = u^*\rho(u),
\label{eq:127}
\end{equation}
so that
\begin{align*}
\omega(\mathfrak{u}) = J u^*\rho(u)J^{-1}[D,
\rho(u)^*u]_\rho = Ju^*J^{-1} [D, \rho(u)^*u]_\rho JuJ^{-1}  = u^\circ\, [D, \rho(u)^*u]_\rho \,{u^*}^\circ.
\end{align*}

The second part of the proposition is obtained turning back to the
definition of $D_{\omega_\rho^u}$, that is the right hand
side of \eqref{eq:19}. One has that $D_{\omega_\rho^u}$ is self-adjoint if and only if 
\begin{equation}
  \label{eq:169}
  \left(\omega_\rho^u - (\omega_\rho^u)^*\right) + \epsilon' J\left(   \omega_\rho^u - (\omega_\rho^u)^*\right) J^{-1}=0.
\end{equation}
By hypothesis $D_{\omega_\rho}=D+\omega_\rho + J\omega_\rho J^{-1}$ is
self-adjoint, so that $(\omega_\rho-\omega_\rho^*)+\epsilon' J\left(\omega_\rho- \omega_\rho^*\right)J^{-1}=0$. Therefore,
  from the definition \eqref{eq:148} of $\omega_\rho^u$, equation \eqref{eq:169} becomes
$ \omega(u) + \epsilon' J \omega(u) J^{-1}=0$
with 
\begin{equation}
\omega(u):=  \rho(u)[D,u^*]_\rho -(\rho(u)[D,u^*]_\rho)^*.
\label{eq:174}
\end{equation}
The result follows remembering that $\rho(u)[D,u^*]_\rho = \rho(u)Du^*
- D$, so that 
\begin{align}
  \label{eq:183}
\omega(u) = \rho(u)Du^*- uD \rho(u)^* &=u \left(u^*\rho(u)D-
                                        D\rho(u)^*u \right)u^*=-u [D, \rho(u)^*u]_\rho\,u^*,
\end{align}
where we used \eqref{eq:127}.
\end{proof}
\begin{rem}
\textup{
  One may check directly the equivalence of the two choices for
  $\omega(u)$ in
  \eqref{eq:177}. Writing
$ \omega:= [D,\rho(u)^*u]_\rho$,
 one gets that for $\omega(u)=  u^\circ \,[D,
  \rho(u)^*u]_\rho \,{u^*}^\circ=  u^\circ\omega\,{u^*}^\circ$, equation \eqref{eq:178bis} is equivalent to 
  \begin{equation}
    \label{eq:171}
    \omega= -\epsilon' (u^\circ)^* \left(J u^\circ \omega\,
      {u^*}^\circ\, J^{-1} \right)(u^{*\circ})^*=-\epsilon' \,u^*Ju \; \omega \;
    u^* J^{-1}u, 
  \end{equation}
where we use that  $u^\circ$ is unitary, with $(u^\circ)^*=
(u^*)^\circ$ so that $(u^{*\circ})^* = u^\circ$, as well as
\begin{equation}
(u^\circ)^* J u^\circ= JuJ^{-1}J Ju^*J^{-1}= \epsilon''
JuJ^{-1}u^*J^{-1}= \epsilon''
u^*J u(J^{-1})^2= u^*Ju,
\label{eq:175}
\end{equation}
and similarly ${u^*}^\circ\, J^{-1} \,u^\circ= u^*
J^{-1}u$. On the other hand, for $\omega(u)=  u\,[D,
  \rho(u)^*u]_\rho \,u^*=u\omega u^*$, equation \eqref{eq:178bis} is equivalent to 
  \begin{equation}
    \label{eq:185}
    \omega = -\epsilon'\, u^* J u \; \omega \; u^* J^{-1} u,
  \end{equation}
which is precisely the r.h.s. of \eqref{eq:171}.
}
\end{rem}
An obvious solution to \eqref{eq:178bis} is that $\rho(u)^*u$
twist-commutes with $D$. This happens in particular when $u$ is
invariant under the twist, $\rho(u)=u$, so that $\rho(u)^*u=\I$.
An extensive study of
\eqref{eq:178bis} and its solutions will be undertaken elsewhere. Here, we just solve it
in the simple example of the minimal twist of manifold.

\subsection{Minimal twist of a manifold}

The minimal twist of a closed spin manifold $\M$ of even dimension $2m$,
$m\in\N$, has been defined in \cite{Landi:2015aa} as the real, graded,
twisted spectral triple 
 \begin{equation}
   \label{eq:184}
\A= \cinf\otimes\C^2, \quad \HH= L^2(\M,S), \quad  D=\ds , \quad J,  \quad \rho
 \end{equation}
where $\cinf$ is the algebra of smooth functions on $\M$, 
the Hilbert space $L^2(\M,S)$ is that of square integrable spinors,
with usual Dirac operator 
\begin{equation}
\ds=-\ii \sum\nolimits_{\mu=1}^{2m}
\gamma^\mu\nabla_\mu\quad \text{ and }\quad\nabla_\mu = \partial_\mu + \omega^\mu
\label{eq:193}
\end{equation}
($\gamma^\mu$ are the Dirac matrices of size $2^m$, $\omega^\mu$
is the spin connection), 
the real structure $J$ is the charge conjugation operator composed with complex conjugation, and the automorphism $\rho$ 
\begin{equation}
\rho(f,g) = (g,f) \quad\quad \forall 
(f,g)\in \A\simeq \cinf\oplus\cinf .
\label{eq:187}
\end{equation}
is the flip. The grading $\Gamma$ (the product of all the Dirac matrices) splits $\HH$  in two
orthogonal subspaces $\HH_{\pm}$, on which each copy
of $\cinf$ acts independently (by point-wise multiplication). The representation $\pi$ of $\A$
on  $\HH = \HH_+ \oplus \HH_-$ is 
\begin{equation}
  \label{eq:186}
  \pi(a)=\left(\begin{array}{cc} f \I_{2^{m-1}}& 0 \\ 0 & g
                                                     \I_{2^{m-1}}\end{array}\right)\quad
                                                 \forall
                                                 a=(f,g)\in\A \text{
                                                   with } f,g\in\cinf.
\end{equation}
Finally, the $KO$-dimension
of the twisted spectral triple \eqref{eq:184} is
$2 m\;\text{mod} \, 8$. 

\begin{prop}
\label{prop:minimaltwist}
  In KO-dim $=0,4$, any unitary of $\A$ is a solution of
  \eqref{eq:178bis}.  On the other hand, in KO-dim $=2,6$, the only
  solutions are the trivial one $\omega(u)=0$.
\end{prop}
\begin{proof}
  A unitary $u$ of $\A$ is (omitting the representation symbol and
  the identity operator)
  \begin{equation}
    u=\left(\begin{array}{cc} e^{\ii\theta_1}& 0 \\ 0 &
                                                             e^{\ii\theta_2} \end{array}\right),
  \end{equation}
where $\theta_1, \theta_2$ are smooth real functions on
$\M$. Hence
\begin{equation}
  \label{eq:189}
  \rho(u)^*u= \left(\begin{array}{cc} e^{-\ii \theta_2}& 0 \\ 0 &
                                                             e^{-\ii \theta_1} \end{array}\right)\,\left(\begin{array}{cc} e^{\ii\theta_1}& 0 \\ 0 &
                                                             e^{\ii\theta_2} \end{array}\right)=\left(\begin{array}{cc} e^{\ii \varphi}& 0 \\ 0 &
                                                             e^{-\ii \varphi} \end{array}\right)
                                                         \quad
                                                         \text{with} \quad\varphi:=\theta_1-\theta_2.
\end{equation}
For any $a\in\A$, one has \cite [eq. (5.9)]{Landi:2015aa}
\begin{equation}
  \label{eq:182}
   [\ds, a]_\rho= -\ii \gamma^\mu (\partial_\mu a),
\end{equation}
so that
\begin{equation}
  \label{eq:197}
   [\ds, \rho(u)^*u]_\rho =-\ii \gamma^\mu\left(
  \begin{array}{cc}
    \partial_\mu e^{\ii\varphi}& 0 \\
    0 &\partial_\mu e^{-\ii \varphi}
  \end{array}\right)=-\ii \gamma^\mu\left(
        \begin{array}{cc} 
          \ii (\partial_\mu\varphi) e^{\ii\varphi}& 0 \\
          0 & -\ii (\partial_\mu\varphi) e^{-\ii \varphi}
        \end{array}\right).
\end{equation}
In addition \cite[eq.(5.10)]{Landi:2015aa},
\begin{equation}
  \label{eq:195}
\gamma^\mu u =\rho(u)\gamma^\mu,
\end{equation}
so by an easy calculation
\begin{equation}
  \label{eq:192}
  u  [\ds, \rho(u)^*u]_\rho u^* = -\ii  \gamma^\mu\left(\begin{array}{cc}  \ii \partial_\mu\varphi& 0 \\ 0 &
                                                             -\ii \partial_\mu\varphi\end{array}\right).
\end{equation}
Furthermore, by \cite[Lem. 5.2]{Landi:2015aa}, for $\omega_\rho= -\ii \gamma^\mu
\rho(a)\partial_\mu b$, one has
\begin{equation}
  \label{eq:194}
  J \omega_\rho J^{-1} = \left\{\begin{array}{ll}
-\ii \gamma^\mu \rho(a^*)\partial_\mu b^* & \text{ if }
                                                KO\mbox{-dim} =0,4,\\ \\
-\ii \gamma^\mu a^*\partial_\mu \rho(b^*)& \text{ if }
                                                KO\mbox{-dim} =2,6.
\end{array}\right.
\end{equation}

\noindent
Therefore, for $KO$-dim $=0,4$ one obtains (remembering that $\varphi$
is a real function)
\begin{align}
\nonumber
   J  u  [\ds, \rho(u)^*u]_\rho u^*  J^{-1} & =
   -\ii \gamma^\mu\left(\begin{array}{cc}  \partial_\mu(\ii\varphi)^*& 0 \\ 0 &
                                                             \partial_\mu(-\ii\varphi)^*\end{array}\right)
                                                                         =  -\ii \gamma^\mu\left(\begin{array}{cc} -\ii \partial_\mu\varphi& 0 \\ 0 &
                                                             \ii \partial_\mu\varphi\end{array}\right)
  \\
\label{eq:resmintwist04}
                                                                        & =-  u  [\ds, \rho(u)^*u]_\rho u^*,
\end{align}

\noindent
whereas for $KO$-dim $=2,6$ one has
\begin{align}
\nonumber
   J  u  [\ds, \rho(u)^*u]_\rho u^*  J^{-1} & = -\ii \gamma^\mu\left(\begin{array}{cc} \partial_\mu(-\ii\varphi)^*& 0 \\ 0 &
                                                              \partial_\mu(\ii\varphi)^*\end{array}\right)=-\ii \gamma^\mu\left(\begin{array}{cc} \ii \partial_\mu\varphi& 0 \\ 0 &
                                                             -\ii \partial_\mu\varphi\end{array}\right)
  \\
\label{eq:resmintwist26}
                                                                        & = u [\ds, \rho(u)^*u]_\rho u^*.
\end{align}
The result follows noticing that in even dimension one has the sign $\epsilon'=1$, so that 
 \eqref{eq:resmintwist04} is solution to \eqref{eq:178bis} for any
 $u$, while \eqref{eq:resmintwist26} is solution only when $\omega(u)=0$. \end{proof}

This simple example exhibits two
interesting cases: the unitaries that preserve the
self-adjointness of the Dirac operator are either the whole group
$C^\infty(\M,U(1)\times U(1))$ of unitaries
of $\A$, or the trivial solution to \eqref{eq:178bis}. Intriguingly, the group $C^\infty(\M,U(1))$ of unitaries which are
invariant under the twist is of no particular
importance. 

To understand why this is the case, recall from \cite[Lemma 5.1]{Landi:2015aa} that in $KO$-dimension
$0,4$, one has $JuJ^{-1}=u^*$, so that $\Ad(u)=\I$.
Therefore, the Dirac
operator is invariant under any gauge transformation, no matter whether
$u$ is invariant under the twist or not.  Moreover, the fact that the
action of $\Ad(u)$ is trivial indicates that the twisted
fluctuations, are \emph{not} generated by the action of a unitary. This fact can be checked
explicitly, computing
$\omega_\rho= \rho(a)[D,a']_\rho$
for $a=(f,g)$, $a'=(f',g')$: one gets from \eqref{eq:182} and \eqref{eq:195}
\begin{equation}
\omega_\rho=\rho(a)[D,a']_\rho= -\ii \gamma^\mu a\,\partial_\mu a' = -
\ii \gamma^\mu\left(\begin{array}{cc} f\, \partial_\mu f'&
                                                              0\\ 0 &
                                                             g\, \partial_\mu
                                                                      g'\end{array}\right),
\label{eq:198}
\end{equation} and by \eqref{eq:194}
\begin{equation}
J\omega_\rho J^{-1} = -\ii \gamma^\mu a^*\partial_\mu a'^*= -
\ii \gamma^\mu\left(\begin{array}{cc}  \bar f\,\partial_\mu \bar f'&0\\ 0 &
                                                             \bar
                                                                            g\,\partial_\mu\bar g'\end{array}\right),
\label{eq:198bis}
\end{equation}
so that
\begin{equation}
\omega_\rho + J\omega_\rho J^{-1} =-\ii \gamma^\mu\left(\begin{array}{cc}  f_\mu&0\\ 0 &
                                                           g_\mu\end{array}\right)
\label{eq:199}
\end{equation}
with $f_\mu =  f\, \partial_\mu f' +\bar f\,\partial_\mu \bar f'$ and 
$g_\mu =  g\, \partial_\mu  g' + \bar g\,\partial_\mu \bar g'$ real
function on $\M$. The r.h.s. of \eqref{eq:199} is self-adjoint if and only if  
\begin{align}
  \label{eq:200}
0&= -\ii \gamma^\mu\left(\begin{array}{cc}  f_\mu&0\\ 0 &
                                                           g_\mu\end{array}\right)
                                                       - \left(-\ii \gamma^\mu\left(\begin{array}{cc}  f_\mu&0\\ 0 &
                                                           g_\mu\end{array}\right)\right)^* \nn \\
                                                       &=-\ii \gamma^\mu\left(\begin{array}{cc}  f_\mu&0\\ 0 &
                                                           g_\mu\end{array}\right)
                                                       - \ii \gamma^\mu\left(\begin{array}{cc}  g_\mu&0\\ 0 &
                                                           f_\mu\end{array}\right)
\nn \\
&=-\ii (f_\mu+ g_\mu)\gamma^\mu,
\end{align}
that is if and only if  $f_\mu=-g_\mu$. In that case, \eqref{eq:199} yields 
$\ds_{\omega_{\rho}}= \ds - \ii f_\mu\gamma^\mu \Gamma$, as already shown
in \cite{Landi:2015aa}.
The point is that such a fluctuation cannot be obtained with $a=u$ a
unitary and $a'=u^*$, that is for $f=e^{\ii \theta_1}$,
$g=e^{\ii \theta_2}$, $f'=e^{-i\theta_1}$, $g'=e^{-i\theta_2}$, since
this would give $f_\mu = g_\mu =0$. 

In $KO$-dimension $2,6$, one has that $\omega(u)=0$ if and only if 
\begin{equation}
[\ds,\rho(u)^*u]_\rho=0.\label{eq:206}
\end{equation}
By \eqref{eq:197} this mean that $u=(e^{\ii \theta_1}, e^{\ii \theta_2})$
with $\theta_1-\theta_2$ a constant function. Notice that this is a
bigger set than the unitaries invariant under the twist (for which the
constant is zero). However, in any case such unitaries do not generate
a fluctuation. Indeed, $\omega_\rho$ is still given by
\eqref{eq:198}, but 
\begin{equation}
  \label{eq:205}
  J\omega_\rho J^{-1} = -
\ii \gamma^\mu\left(\begin{array}{cc}  \bar g\,\partial_\mu \bar g'&0\\ 0 &
                                                             \bar
                                                                            f\,\partial_\mu\bar f'\end{array}\right).
\end{equation}
Thus $\omega_\rho + J\omega_\rho J^{-1}$ is given by \eqref{eq:199}
with
\begin{equation}
f_\mu =  f\, \partial_\mu f' +\bar g\,\partial_\mu \bar g',\qquad 
g_\mu=\bar f_\mu
.\label{eq:208}
\end{equation}
With $f=e^{\ii \theta_1}$, $g=e^{\ii \theta_2}$,
$f'=e^{-i\theta_1}$, $g=e^{-i\theta_2}$, one gets $f_\mu =
\ii \, \partial_\mu(\theta_1-\theta_2)$, which vanishes when
$\theta_1-\theta_2$ is constant. More generally, one finds back the
result of \cite{Landi:2015aa} noticing that for arbitrary $f,f'$ and $g,g'$,
a computation similar to \eqref{eq:200} yields that 
$\omega_\rho + J\omega_\rho J^{-1}$ is self-adjoint if and only if  $f_\mu=g_\mu
=0$.

To summarize, one has the following result.
\begin{prop}
 In  $KO$-dimension
$0,4$, the operator $\ds$ has non-zero twisted self-adjoint fluctuations given by
\begin{equation}
\ds_{\omega_\rho}= \ds -\ii f_\mu\gamma^\mu\Gamma , \qquad f_\mu \in C^{\infty}(\M, \R).
\label{eq:204}
\end{equation}
They are invariant under a gauge transformation, but are not
generated by the action of unitaries. 
In  $KO$-dimension $2,6$, there is no non-zero self-adjoint fluctuations. 
\end{prop}

\appendix

\section{The non-twisted case} 


The material in this Appendix is well known and  taken mainly from
\cite{Connes:1996fu} and \cite{Connes-Marcolli}. 

\subsection{Fluctuations and Morita equivalence} \label{subsec:gauge-non-twist}

Recall that a finitely generated, projective (right, say) $\A$-module $\E$ is hermitian if it comes equipped 
with an $\A$-valued inner product, that is a sesquilinear map $\rhs{\cdot}{\cdot}: \E \times \E \to \A$
such that $\rhs{\xi}{\xi} \geq 0$ for any $\xi\in\E$, $(\rhs{\xi}{\eta})^*=\rhs{\eta}{\xi}$ and  
$\rhs{\xi a}{\eta b} = a^* \rhs{\xi}{\eta} \,b$, for all $\xi, \eta\in\E$ and $a,b\in\A$.
A similar notion goes for left-modules with a sesquilinear map $\lhs{\cdot}{\cdot}: \E \times \E \to \A$ which is now 
linear in the first entry (and anti-linear on the second). 
The module $\E$ is taken to be self-dual for the $\A$-valued hermitian structure \cite[Prop.~7.3]{ri10}, in the sense that for any 
$\varphi \in \Hom_{\A}(\E, \A)$ there exists a unique $\zeta_\varphi\in \E$ such that 
$\varphi(\xi)=\rhs{\zeta_\varphi}{\xi}$, for all $\xi\in\E$.

In the crudest version \cite{ri10}, the algebra $\B$ is
 Morita equivalent to the (unital) algebra $\A$ if there
exists a hermitian finite projective $\A$-module $\E$
such that $\B$ is isomorphic to the algebra $\End_\A(\E)$ 
of $\A$-linear endomorphisms of $\E$ which are adjointable (with respect to the hermitian structure of 
$\E$).
In particular an algebra is Morita equivalent to itself. 
In that case the module $\E$ can be taken to be the algebra itself, with 
hermitian map $\rhs{a}{b} = a^* b$ or $\lhs{a}{b} = a b^*$.

%

\subsubsection{Morita equivalence by right module}
\label{subsec:Moritarightnontwist}

Let us  assume that the module implementing the Morita equivalence
between $\A$ and $\B$ is a right $\A$-module $\E_R$ with $\A$-valued inner product $\rhs{\cdot}{\cdot}$. 
  The action  of $\B \simeq \End_\A(\E)$
on $\E_R$ is
not suitable to build a spectral triple, for $\E_R$ is not an
Hilbert space. However, the tensor product
\begin{equation}
\HH_R := \E_R\otimes_\A \HH
\label{eq:45}
\end{equation}
 is an
Hilbert space for the inner product \cite[p. 562]{Connes:1994kx}
\begin{equation}
  \label{eq:32}
\hs{\eta_1\ot \psi_1}{\eta_2\ot \psi_2}_{\HH_R} = 
\hs{\psi_1}{ \rhs{\eta_1}{ \eta_2} \psi_2}_{\HH}
  \qquad \forall \,\eta_1,\eta_2\in\E_R,\;
  \psi_1, \psi_2\in\HH,
\end{equation}
with $\langle \cdot, \cdot \rangle_\HH$ the inner
product of $\HH$. 
The action  of $\B \simeq \End_\A(\E)$ is then extended to $\HH_R$ as
\begin{equation}
\pi_R(b)(\eta\otimes \psi) := b\eta\otimes \psi \qquad \forall b\in{\cal
  B}.
\label{eq:44}
\end{equation}
To make $D$ act on $\HH_R$, the simplest guess, 
\begin{equation}
D_R(\eta\otimes
\psi) := \eta\otimes D\psi,
\label{eq:83}
\end{equation}
is not compatible with the tensor product
over $\A$ \cite[p. 204]{Connes-Marcolli}
since
\begin{equation}
D_R(\eta a\otimes \psi)   - D_R(\eta\otimes a\psi) = \eta a\otimes
D\psi -\eta\otimes
Da\psi= -\eta\otimes [D,a]\psi
\label{eq:34}
\end{equation}
has no reason to vanish.  To correct this, one uses the derivation
$\delta=[D,\cdot \, ]$ of $\A$ in the $\A$-bimodule $\Omega^1_D(\A)$ of $1$-forms as defined in \eqref{eq:71}. Since both $\Omega^1_D(\A)$ and $\A$ act on $\HH$
as bounded operators  in a compatible way \eqref{eq:129}, 
the r.h.s. of \eqref{eq:34}, viewed as $-(\eta\otimes \delta(a))\psi$, is made  zero by adding to $D_R$ an
$\Omega^1_D(\A)$-valued connection $\nabla$ on $\E$. One thus defines
the gauged operator  
\begin{equation}
  \label{eq:29}
  D_R(\eta\otimes \psi):= \eta\otimes D\psi +
  (\nabla\eta)\psi\quad\quad\forall \eta\in\E_R, \psi\in\HH,
\end{equation}
and checks by Proposition  \ref{Prop:lebnizright} that this is linear, since
\begin{align}
D_R(\eta a\otimes \psi)
- D_R(\eta\otimes a\psi)&=D_R(\eta a\otimes \psi - \eta\otimes a\psi)
+ \nabla(\eta a) \psi - (\nabla\eta)a\psi, \nonumber \\
&=-\eta \otimes [D,a]\psi + \eta\otimes \delta(a)\psi=0.
\label{eq:121}
\end{align}

If the right $\A$-module $\E_R$ is finite projective thus of the type $\E_R=p\A^N$ for some $N\in\N$, with $p$ a
self-adjoint matrix in $M_N(\A)$ such that $p^2=p$. Moreover,  given a derivation $\delta$ of
$\A$ in a $\A$-bimodule $\Omega$, 
any $\Omega$-valued connection 
is of the form
\begin{equation}
  \nabla = \nabla_0 + \bo
\label{eq:48}
  \end{equation}
where 
\begin{equation}
 \nabla_0\, \eta  =p\left(\begin{array}{c} \delta(\eta_1)\\ \vdots\\
                            \delta(\eta_N)\end{array}\right)\quad\quad
                        \forall \, \eta = p \left(\begin{array}{c} \eta_1\\ \vdots\\
                            \eta_N\end{array}\right)\in\E_R, \quad
                        \eta_j\in \A,
\end{equation}
is the Grassmann connection, while $\bo$ is an $\A$-linear map $\E_R\to\E_R\otimes\Omega$, that is 
\begin{equation}
  \label{eq:132}
  \bo(\eta a) = \bo(\eta)\cdot a\quad \forall \eta\in\E_R, a\in\A.
\end{equation}
In particular, for a self-Morita equivalence the operator $D_R$ has a friendlier form.
\begin{prop}
\label{recall:D_R}
 For $\B=\A$ and $\E_R=\A$, 
  one obtains
  \begin{equation}
    \label{eq:160}
    D_R = D+ \omega \qquad \mbox{for some} \quad \omega\in\Omega^1_D(\A).
  \end{equation}
\end{prop}
\begin{proof}
Any connection $\nabla$ on $\E_R=\A$ is written as $\nabla = \delta + \omega$ for a  
$1$-form $\omega \in\Omega^1_D(\A)$. 
Then,  
\begin{equation}
    \label{eq:53}
    D_R(a\otimes \psi) = a \otimes D\psi + \I \otimes (\delta(a) + \omega a) \psi. 
  \end{equation}
  Identifying $a\otimes \psi\in \HH_R$ with $a\psi\in \HH$, one
  rewrites \eqref{eq:53} as
  \begin{equation}
    \label{eq:65}
    D_R(a\psi) = aD\psi + \delta(a)\psi + \omega a\psi = aD\psi + (D a - a D)\psi + \omega a\psi = (D+\omega)(a\psi),
  \end{equation}
  meaning that the action of $D_R$ on $\HH$ coincides with the
  operator $D+\omega$.
\end{proof}

 Since $\omega$ is bounded, the operator $D_R$  has a compact resolvent  
 and bounded commutator
with $\A$. Consequently, {for a self-adjoint $\omega$ one gets that 
\begin{equation}
(\A, \HH, D +\omega)
\label{eq:110}
\end{equation}
is a spectral triple \cite{BMvS16}, \emph{Morita equivalent} to $(\A, \HH, D)$. Furthermore, any grading $\Gamma$ of  
$(\A, \HH, D)$, since anticommutes with any $a[D,b]$, hence with $\omega$, thus with $D_R$,  is also a grading of $(\A, \HH, D_R)$. 

However, if $(\A, \HH, D)$ is a real spectral triple with real structure $J$, the later
is not necessarily a real structure for  \eqref{eq:110}. Indeed
$J(D+\omega) =\epsilon'(D+\omega)J$ if and only if $\omega= \epsilon'
J\omega J^{-1}$. This
has no reason to be true, because of the following lemma (whose proof follows from \eqref{eq:4}, \eqref{eq:24}).
\begin{Lemma}
\label{Lemma:fluctreal}
   Let $(\A, \HH, D), J$ be a real spectral triple, and  $\omega=\sum\nolimits_j a_j[D,b_j]\in\Omega^1_D(\A)$. Then   \begin{equation}
    \label{eq:61}
    J\omega J^{-1} = \epsilon '  \Big(\sum\nolimits_j (a_j^*)^\circ[D,(b_j^*)^\circ] \Big).
  \end{equation}
\end{Lemma}

\subsubsection{Morita equivalence by left module}
\label{subsec:Moritleft}

To implement $\A$ self-Morita equivalence in a way compatible
with the real structure, one uses $\A$ not only as a right
$\A$-module $\E_R$, but also as a left $\A$-module $\E_L$ (as
explained in this section), then as a
$\A$-bimodule $\E$ (this is the content of \S\ref{subsec:bimod}).

 In defining the Hilbert space $\HH_R$ in \eqref{eq:45}, one takes advantage of the left $\A$-module
structure of $\HH$ induced by the representation $\pi$. Alternatively,
 one has available the right
 $\A$-module structure \eqref{eq:35} of $\HH$, $\psi a = a^\circ \psi$ for $\psi\in\HH, a\in\A$, which offers a
 possibility to implement the
Morita equivalence between $\A$ and $\B$ thanks to a hermitian
finite projective \emph{left} $\A$-module $\E_L$, with $\A$-valued  inner product $\lhs{\cdot}{\cdot}$. One thus
considers the Hilbert space
\begin{equation}
  \label{eq:54}
  \HH_L:= \HH\otimes_\A \E_L,
\end{equation}
with inner product
\begin{equation}
\langle \psi_1\otimes\eta_1, \psi_2\otimes\eta_2\rangle_{\HH_L}
= \langle \psi_1 \, \lhs{\eta_1}{\eta_2} ,\psi_2 \rangle_\HH.\label{eq:176}
\end{equation}
The right action of $\B \simeq \End_\A(\E)$ on $\E_L$  is extended
to $\HH_L$ as 
\begin{equation}
  \label{eq:55}
(\psi\otimes \eta) \, b := \psi \otimes \eta b.
\end{equation}
Again,  the natural action
\begin{equation}
D_L(\psi\otimes \eta):=D\psi\otimes\eta
\label{eq:82}
\end{equation}
of $D$ on $\HH_L$ is not compatible
with the tensor product over $\A$ because
\begin{align}
 \nonumber
  D_L(\psi \otimes a\eta) -D_L(\psi a\otimes\eta)&=  (D\psi)\otimes
   a\eta  - D(\psi
  a)\otimes \eta   = (D\psi)a\otimes \eta - D(\psi
  a)\otimes \eta ,\\
 \label{eq:56bis}& = a^\circ(D\psi)\otimes \eta  -D(a^\circ\psi
  )\otimes \eta  =- [D , a^\circ]\psi
  \otimes \eta
\end{align}
does not vanish. To correct this, one uses a connection
$\nabla^\circ$ on $\E_L$ with value in the $\A$-bimodule 
\begin{equation}
\Omega^1_D(\A^\circ)= \Big\{ \sum\nolimits_j a_j^\circ [D, b_j^\circ],\quad a_j^\circ, b_j^\circ\in\A^\circ\Big\}
\label{eq:69}
\end{equation} 
generated by the  derivation
\begin{equation}
\delta^\circ(a):= [D, a^\circ],
\label{eq:97}
\end{equation}
with bimodule law
\begin{equation}
  a\cdot \omega^\circ\cdot   b := b^\circ \omega^\circ a^\circ .
\label{eq:74}
\end{equation}
This law guarantees that \eqref{eq:69} is indeed a bimodule over $\A$
and $\delta^\circ$ a derivation of $\A$ (not
of $\A^\circ)$ with values in $\Omega^1_D(\A^\circ)$. The relation between $\Omega^1_D(\A^\circ)$ and
$\Omega^1_D(\A)$ is given by the following lemma, whose proof follows
from  \eqref{eq:4} and  \eqref{eq:24}.
\begin{Lemma}
  \label{actionoppositeform}
Any $\omega^\circ= \sum\nolimits_j a_j^\circ [D, b_j^\circ]$ in
$\Omega^1_D(\A^\circ)$ acts on the left on $\HH$ as the
bounded operator
\begin{equation}
\omega^\circ = \epsilon' J\omega J^{-1}\label{eq:162}
\end{equation}
for $\omega=\sum\nolimits_j a_j^*[D,b_j^*] \in\Omega^1_D(\A)$.
\end{Lemma}

  The right action of $\A$ and the left action of
  $\Omega^1_D(\A^\circ)$ on $\HH$ (corresponding to a right action of $\Omega^1_D(\A)$) are compatible in the sense of condition
  \eqref{eq:130}: 
  \begin{align}
    \label{eq:43}
    (a\cdot\omega^\circ)\psi&= (\omega^\circ a^\circ)\psi =\omega^\circ(\psi a).
  \end{align}
  The connection $\nabla^\circ$ thus defines an operator
  $\HH_L\to\HH_L$ which satisfies the Leibniz rule \eqref{eq:123b}, therefore the following 
    is a well defined operator on $\HH_L$,
  \begin{equation}
    \label{eq:57}
    D_L(\psi\otimes \eta) := D\psi\otimes \eta + \nabla^\circ(\psi \otimes \eta) .
  \end{equation}

For a left module $\E_L \simeq A^Np$ with $p=p^2\in M_N(\A)$ 
the connection decomposes as
\begin{equation}
\nabla^\circ=\nabla^\circ_0 + \bo^\circ
\label{eq:133}
\end{equation}
with Grassmann connection
\begin{equation}
  \label{eq:50}
  \nabla^\circ_0 \,\eta = (\delta^\circ(\eta_1), \hdots, \delta^\circ(\eta_N))\,p \quad\quad
  \forall\, \eta=(\eta_1, \hdots, \eta_N)\in \E_L, \quad \eta_j \in \A ,
\end{equation}
while $\bo^\circ$ is a  map $\E_L\to
\Omega^1_D(\A^\circ)\otimes_\A \E_L$ which is $\A$-linear in the sense that
\begin{equation}
  \label{eq:134}
  \bo^\circ (a\eta) = a\cdot \bo^\circ(\eta).
\end{equation}
We use this to get a more tractable expression for $D_L$, for a self-Morita equivalence. 
\begin{prop}
\label{prop:Moritaleft}
 For $\B=\A$ and $\E_L=\A$, the construction above yields
  \begin{equation}
    \label{eq:158}
    D_L = D + \omega^\circ = D+\epsilon' J \omega J^{-1}
  \end{equation}
for some $\omega^\circ=\epsilon'J\omega J^{-1}\in \Omega^1_D(\A^\circ)$,  with $\omega\in\Omega^1_D(\A)$.
\end{prop}
\begin{proof}
The decomposition \eqref{eq:133} will now read $\nabla^\circ = \delta^\circ + \omega^\circ$ with the form 
$\omega^\circ$ be such that 
  \begin{equation}
    \label{eq:105}
    \bo^\circ(a)
    = a\cdot \omega^\circ = (a\cdot \omega^\circ) = \omega^\circ a^\circ .
  \end{equation}
Therefore;
  \begin{align}
    \label{eq:131}
D_L(\psi\otimes a)  &= D\psi \otimes a +\delta^\circ(a)\psi\otimes \I +
                          \omega^\circ a^\circ \psi\otimes \I , \nonumber \\
                        &= (D\psi) a \otimes\I+ (Da^\circ - a^\circ D)\psi \otimes \I +
                          \omega^\circ a^\circ \psi\otimes \I, \nonumber \\
                        &= a^\circ D\psi \otimes\I + (Da^\circ - a^\circ D)\psi \otimes \I +
                          \omega^\circ a^\circ\psi\otimes \I \nonumber \\
                          &= D (a^\circ\psi) \otimes \I +
                          \omega^\circ a^\circ \psi\otimes \I.
  \end{align}
  Identifying $a^\circ\psi\otimes\I = \psi \otimes a = \psi a \otimes \I$ in $\HH_L$ with
  $a^\circ\psi \in\HH$, one obtains that $D_L$ acts as
  $D+\omega^\circ$. The rest of the result follows from Lemma \ref{Lemma:fluctreal}.
\end{proof}
As in the right module case, when $\omega$ is self-adjoint  the datum 
\begin{equation}
(\A, \HH, D + \epsilon J\omega J^{-1})\label{eq:30}
\end{equation}
is a spectral triple,
admitting as grading any grading of $(\A, \HH, D)$.
However it is not a real spectral triple for the real structure $J$,
because $J(D+ \epsilon' J\omega J^{-1}) = (D+ \epsilon' J\omega J^{-1})J$ if and
only if $\omega= \epsilon' J\omega J^{-1}$. This has no reason to be true, by Lemma \ref{Lemma:fluctreal}.

\subsubsection{Morita equivalence by bimodule and the real structure}
\label{subsec:bimod}
To make the real structure compatible with Morita equivalence of
spectral triples, one needs to combine the two constructions above. Explicitly, given a real spectral triple $(\A, \HH,
D)$, one first implements the self-Morita equivalence
of $\A$ by using the right module $\E_R=\A$ --- thus obtaining the spectral
triple \eqref{eq:110}; then uses this with the left module $\E_L=\A$. This
yields the Morita equivalent 
spectral triple $ (\A, \HH, D' )$
where 
\begin{equation}
  \label{eq:33}
 D' =D + \omega_R + \epsilon' J\omega_L J^{-1}
\end{equation}
with $\omega_R$, $\omega_L$  two self-adjoint elements of $\Omega^1_D(\A)$, a priori
distinct. It  is the real structure that forces these two $1$-forms to
be equal.
\begin{prop}
\label{prop:fluctbimodule}
The real structure $J$ of $(\A, \HH, D)$ is a real structure for the
Morita equivalent spectral triple $(\A, \HH, D' )$ if
 and only if there exists $\omega\in\Omega^1_D(\A)$ such that 
 \begin{equation}
   D' = D_\omega:= D + \omega + \epsilon'J\omega J^{-1}.
 \end{equation}
\end{prop}
\begin{proof}
By an easy computation, one finds that $JD'= \epsilon'DJ$ if and only if
\begin{equation}
  \label{eq:46}
 (\omega_L- \omega_R) - \epsilon' J (\omega_L - \omega_R)J^{-1} =0.
\end{equation}
Adding half of this expression to the r.h.s. of \eqref{eq:33}, one gets
\begin{equation}
  \label{eq:64}
  D' = D + \tfrac 12(\omega_R +\omega_L) + \epsilon' J \tfrac 12(\omega_R + \omega_L) J^{-1}.
\end{equation}
Hence the result with $\omega:= \tfrac 12(\omega_R + \omega_L)$.
\end{proof}
\begin{rem}\label{fluctgrassmann}
  \textup{
  Taking as a connection the Grassmann connection in the definition
  \eqref{eq:29} of $D_R$ (i.e. $\nabla=\nabla_0$), one finds that $D_R$
  coincides with $D$. Similarly, taking $\nabla^\circ=\nabla^\circ_0$
  in \eqref{eq:57} yields $D_L=D$. Then $D'$ in  \eqref{eq:33}
  coincide with $D$ as well. 
In other terms, given a real spectral triple $(\A, \HH, D)$, implementing the self-Morita equivalence of $\A$ using
the Grassmann connection on the $\A$-bimodule $\A$ leaves the Dirac
operator invariant (i.e. it fluctuates with $\omega=0$).
}\end{rem}

\subsection{Gauge transformations}
\label{sec:annexgauge}

Also the material in this section is well known and mainly taken from
\cite{Connes:1996fu} and \cite{Connes-Marcolli}. 

\subsubsection{Gauge transformations on a hermitian module}
\label{subsec:annexgaugemodule}

An endomorphisms $u\in\End_\A(\E)$  of a hermitian $\A$-module
$\E$ is unitary if $u^*u=uu^* = \id_\E$, where
the adjoint of an operator is defined using the hermitian structure by
\begin{equation}
  \label{eq:98}
  \langle T^*\eta, \xi\rangle := \langle \eta, T\xi\rangle , \qquad
  \forall T\in\End_\A(\E),\, \xi,\eta\in\E.
\end{equation}
Unitary endomorphisms form a group ${\cal U}(\E)$,
acting on $\Omega$-valued connections on $\E$ as
  \begin{equation}
    \label{eq:126}
    \nabla^u := u\nabla u^* \quad\quad \forall u\in {\cal U}(\E),
  \end{equation}
where ${\cal U}(\E)$ acts on
$\E\otimes\Omega$ (if $\E$ is a right $\A$-module) or $\Omega\otimes
\E$  (if $\E$ is a left $\A$-module) as
\begin{equation}
  \label{eq:103}
  u \ot \id_{\Omega} ,
  \quad\quad \mbox{or} \quad\quad 
 \id_{\Omega} \ot \, u .
\end{equation}
Not surprisingly, such an action is a \emph{gauge transformation}.
\begin{prop}
The operator $\nabla^u$ is a connection, for any $u\in{\cal U}(\E)$ and connection $\nabla$.
 \end{prop}
\begin{proof}
In case $\E$ is a right $\A$-module,  one gets from \eqref{eq:103} and \eqref{eq:18} that
\begin{equation}
  \label{eq:140}
  u(\nabla(\eta)\cdot a)=
(u\nabla(\eta))\cdot a \, .
\end{equation}
Hence
  \begin{align*}
    \nabla^u(\eta a) &= u\nabla(u^*(\eta a))= u\nabla(u^*(\eta) a) =
                       u\left( \nabla(u^*(\eta)).a +u^*(\eta)\otimes \delta(a) \right),  \\
                     &= (u\nabla(u^*(\eta))).a +\eta\otimes \delta(a) =\nabla^u(\eta).a +\eta\otimes \delta(a),
  \end{align*}
showing that $\nabla^u$ is a connection. For a left
$\A$-module $\E$ one has from \eqref{eq:67}
\begin{equation}
  \label{eq:141}
   u(a\cdot\nabla(\eta))=
a\cdot u(\nabla(\eta)),
\end{equation}
so that
  \begin{align*}
    \nabla^u(a\eta) &= u\nabla(u^*(a\eta))= u\nabla(au^*(\eta)) =
                       u\left(a\cdot\nabla(u^*(\eta)) + \delta(a)\otimes u^*(\eta) \right),\\
                     &= a\cdot u\nabla(u^*(\eta)) +\delta (a)\otimes
                       \eta =\nabla^u(\eta).a +\delta(a)\otimes \eta.
  \end{align*}
 Hence the result.
\end{proof}

With $\nabla_0$ the Grassmann connections an $\boldsymbol{\omega}$ defined in \eqref{eq:48} or \eqref{eq:133}, 
any connection
\begin{equation}
\nabla=\nabla_0+ \boldsymbol{\omega}
\label{eq:138}
\end{equation}
is  mapped under a gauge transformation to 
\begin{equation}
\nabla^u=\nabla_0+ {\boldsymbol \omega}^u .
\label{eq:142}
\end{equation}
with the gauge transformation fully encoded in the law of transformation of the \emph{gauge potential}
\begin{equation}
{\boldsymbol \omega}\to {\boldsymbol\omega}^u .
\label{eq:136}
\end{equation}
Explicitly, given a right (or left) $\A$-module $\E_R=p\A^N$
(or $\E_L=\A^Np$), a unitary endomorphism is a unitary matrix in
$M_N(\A)$ that commutes with $p$, 
\begin{equation}
{\cal U}(\E_{L,R}):=\left\{u\in M_N(\A) ,\, [u,p]= 0, \, u^*u= \id_\E \right\},
\label{eq:94}
\end{equation} 
and acts by ordinary matrix multiplication
\begin{equation}
  \label{eq:128}
  u(\eta) := p (u\eta) \; \text{ for }\; \eta \in \E_R, \quad\quad 
u(\eta) := (\eta u^*) p  \;\text{ for }\; \eta\in\E_L.
\end{equation}
The choice to act with $u^*$ instead of $u$ in the left-module case is
discussed in Remark \ref{rem:actionunitaire}. 

Given a derivation
$\delta$ of $\A$, we denote by
$\delta(u), \delta(u^*)$ the
elements of $M_N(\Omega)$ with components $\delta(u_{ij})$ or  
$\delta(u^*_{ij})\in\Omega$, $1\leq i,j\leq N$, where $u_{ij}, u^*_{ij}\in\A$
are the components of $u, u^*$. 
\begin{prop}
The gauge transformations on right and left modules are given by 
\begin{align}
  \label{eq:139}
  {\boldsymbol \omega}^u(\eta) &:= p\,u\cdot\delta(u^*)\cdot\eta + u({\boldsymbol \omega}(u^*(\eta))) \qquad \forall \eta \in \E_R,\\
  \label{eq:139bis}
  {{\boldsymbol \omega}}^u(\eta) &:= \eta \cdot\delta(u)\cdot u^* \,p + u({{\boldsymbol \omega}}(u^*(\eta))) \qquad\forall \eta\in \E_L.
\end{align}
\end{prop}
\begin{proof}
For $\eta = p (\eta_j) \in\E_R$ (with $\eta_j\in\A$), using that $p$
commutes with $u^*$ and $p\eta=\eta$, one gets 
  \begin{align*}
    \nabla_0(u^*(\eta)) &=     \nabla_0(pu^*\eta)  =
    \nabla_0(u^*p\eta)  =     \nabla_0(u^*\eta) \nonumber \\ &= p \left(\begin{array}{c}\delta(u^*_{1j}
                                       \eta_j)\\ \vdots \\ \delta(u^*_{Nj}
                                       \eta_j)\end{array}\right)
 = p
                                   \left(\begin{array}{c}\delta(u^*_{1j})\cdot\eta_j
                                           + u^*_{1j} \cdot\delta(\eta_j)\\ \vdots \\\delta(u^*_{Nj})\cdot
                                           \eta_j +
                                           u^*_{Nj}\cdot\delta(\eta_j)\end{array}\right) \nonumber \\ & =
                                       p\delta(u^*)\cdot\eta + u^*\nabla_0 (\eta) ,
                                     \end{align*}
with summation on the
index $j=1,..., N$.
Acting with $u$ on the left, one gets
\begin{equation}
  \label{eq:1288}
  \nabla^u_0 = \nabla_0 + p\,u\cdot \delta(u^*).\end{equation}
the result follows from \eqref{eq:138}, \eqref{eq:142}.
Similarly,  for $\eta\in\E_L$,
  one has
 \begin{align*}
\nonumber
    \nabla_0(u^*(\eta)) &= \left(\delta(
                                       \eta_j u_{j1}), \hdots,\delta(
                                       \eta_j u_{jN})\right) p \\ &=
                                   \left(\delta(\eta_j)\cdot u_{j1}
                                           + \eta_j\cdot\delta(u_{j1} ), \hdots, \delta(\eta_j)\cdot u_{jN} +
                                          \eta_j\cdot\delta(u_{jN})\right)
                                      p\\
&=
                                      \nabla_0
                                       (\eta)\cdot u+  \eta\cdot \delta(u)p.
                                     \end{align*}
Acting with the endomorphism $u$ on the left, which by \eqref{eq:128} amounts to
multiply by the matrix $u^*$ on the right, one obtains
\begin{equation}
  \label{eq:144}
  \nabla_0^u(\eta) = \nabla_0(\eta) + \eta\cdot \delta(u)\cdot u^*\,p.
\end{equation}
 Hence the result.
               \end{proof}

\subsubsection{Gauge transformation for a spectral triple}
\label{subsec:annexgaugetriple}

Let $(\A, \HH, D)$ be a real spectral triple, and consider the right $\A$-module $\E_R=\A$, with derivation $\delta(\cdot)=[D, \cdot]$ in the
$\A$-bimodule $\Omega^1_D(\A)$ defined in \eqref{eq:71}. The
equation \eqref{eq:139} yields the usual law of transformation of the gauge potential, 
\begin{equation}
  \label{eq:143}
  \omega^u  = u[D,u^*] + u\,\omega\, u^*.
\end{equation}
Under a gauge transformation, the gauged Dirac operator $D_\omega$ in \eqref{eq:77} is thus mapped to
\begin{equation}
  \label{eq:145}
  D_{\omega^u}= D + \omega^u + \epsilon' J \omega^u J^{-1}.
\end{equation}
\begin{rem}\label{rem:actionunitaire}
\textup{
To write \eqref{eq:145}, one applies the gauge transformation
$\omega\to\omega^u$ on the
operator $D_\omega$ obtained in Proposition \ref{prop:fluctbimodule}, that is once $\omega_L$ and $\omega_R$ have
been identified. \\
For the sake of coherence, let us check that the
same result follows by applying the gauge transformation on $\omega_L$
and $\omega_R$ independently. Consider the left module $\E_L=\A$ with derivation $\delta^\circ(a)=  [D,a^\circ]$ in $\Omega^1_D(\A^\circ)$ defined in \eqref{eq:69}. By Lemma \ref{actionoppositeform}, a  gauge
  potential in $\Omega^1_D(\A^\circ)$ is  $\omega^\circ=\epsilon' J \omega_L
  J^{-1}$ with $\omega_L\in\Omega^1_D(\A)$. The law of transformation  \eqref{eq:139bis} 
 reads
  \begin{align*}
{\omega^\circ}^u &=   \delta^\circ(u)\cdot u^* + u\cdot
                   \omega^\circ\cdot u^* =  {u^*}^\circ \delta^\circ(u)+ {u^*}^\circ \omega^\circ {u}^\circ,\\
& = {u^*}^\circ[D, {u}^\circ] + {u^*}^\circ \omega^\circ {u}^\circ =\epsilon' J u[D,u^*]J^{-1} + \epsilon' J\,u \,\omega_L u^* J^{-1}= \epsilon' J\,\omega_L^{u}\, J^{-1}.
  \end{align*}
Thus the operator $D+\omega_R+ \epsilon' J \omega_L J^{-1}$ in
Proposition  \ref{prop:fluctbimodule} is mapped under a gauge transformation
to $D + \omega_R^u +   \epsilon' J \omega^u_L J^{-1}$, meaning that
$\omega=\tfrac 12(\omega_R+\omega_L)$ is mapped to
\begin{equation}
\tfrac 12(\omega^u_R+\omega^u_L) = u[D, u^*] + u \frac
12(\omega_R+\omega_L)u^* = u[D, u^*] + u \omega u^*.
\label{eq:73}
\end{equation}
One thus finds back \eqref{eq:143}, as expected. 
}
\end{rem}

Remarkably \cite{Connes:1996fu}, the gauge transformation $D_\omega \to D_{\omega^u}$ can be
retrieved from the adjoint action on $\HH$ of the unitary group of
$\A$, defined by using the real structure. That is, for any unitary element $u\in\A$, $u^*u= uu^*=\I$,
one defines
\begin{equation}
  \label{eq:106}
  \Ad(u)\psi := u \,\psi \,u^*  = u J  u J^{-1}\psi\quad \forall \psi\in\HH.
\end{equation}
Under this action, the Dirac operator is mapped to
$\Ad(u)\, D \Ad(u)^{-1}$. By the order zero and the 
first order conditions,  one shows that \cite[Prop. 1.141]{Connes-Marcolli}
\begin{align}
  \label{eq:11}
\Ad(u)\, D \Ad(u)^{-1} =
  D + u[D,  u^*]  + \epsilon' Ju[D, u^*]J^{-1},
\end{align}
which is nothing but the operator $D_{\omega^u}$ of \eqref{eq:145} obtained
for $\omega=0$ so that $\omega^u=u[D,  u^*]$ from \eqref{eq:143}. 
More generally, 
for a gauged Dirac operator 
\begin{equation}
D_\omega= D + \omega + \epsilon'  \, J \omega  J^{-1}\label{eq:18bis}
\end{equation}
where $\omega$  is an arbitrary self-adjoint element of 
$\Omega^1_D(\A)$, one has \cite[Prop. 1.141]{Connes-Marcolli})
\begin{equation}
  \label{eq:5}
  \Ad(u) \, D_\omega\, \Ad(u)^{-1} = D_{\omega^u}
\end{equation}
with $\omega^u$ defined in \eqref{eq:143}.

\end{document}